\documentclass[11pt]{article}

\usepackage{odonnell}
\usepackage{graphicx}
\usepackage{subfigure}
\usepackage{epstopdf}
\usepackage{braket}

\newcommand{\rnote}[1]{}
\newcommand{\jnote}[1]{}

\newcommand{\unitary}[1]{\mathrm{U}(#1)}
\newcommand{\gl}[1]{\mathrm{GL}_d}
\newcommand{\glrep}[1]{\pi_{#1}} \newcommand{\uirrep}{\pi}
\newcommand{\symm}[1]{{\mathfrak S}(#1)}

\newcommand{\symdim}[1]{\dim(#1)}

\newcommand{\SW}[2]{\mathrm{SW}^{#1}(#2)}

\newcommand{\rsk}{\mathrm{RSK}}
\newcommand{\shRSK}{\mathrm{shRSK}}

\newcommand{\keyl}[2]{\mathrm{K}_{#1}(#2)}

\newcommand{\rank}{\mathrm{rank}}
\newcommand{\diag}{\mathrm{diag}}
\newcommand{\spec}{\mathrm{spec}}

\newcommand{\power}[1]{\Delta_{#1}}

\newcommand{\dtvk}[3]{d_{\mathrm{TV}}^{(#1)}(#2,#3)}

\newcommand{\SWdist}[2]{\mathrm{SW}^{#1}(#2)}

\newcommand{\RSK}{\mathrm{RSK}}
\newcommand{\SYT}{\mathrm{SYT}}
\newcommand{\sh}{\mathrm{sh}}
\newcommand{\LIS}{\mathrm{LIS}} 

\newcommand{\high}{\mathrm{ht}}
\newcommand{\behead}{\mathrm{behead}}
\newcommand{\beheaded}{beheaded\xspace}
\newcommand{\beheading}{beheading\xspace}
\newcommand{\beheadings}{beheadings\xspace}
\newcommand{\curtail}{\mathrm{curtail}}
\newcommand{\curtailment}{curtailment\xspace}

\newcommand{\ssLIS}{\mathbin{\rhd\!\!\!\!\gg}}

\newcommand{\Ham}[3]{{#1}^{#2}_{#3}}
\newcommand{\SymHam}[3]{{#1}^{#2}_{#3,#2-#3}}

\newcommand{\dd}{\,..\,}

\newcommand{\bm}{\boldsymbol{m}}

\newcommand{\ullam}{\underline{\lambda}}
\newcommand{\ulblam}{\underline{\blambda}}

\newcommand{\upperp}[2]{{#1}^{[#2]}}
\newcommand{\lowerp}[2]{{#1}_{[#2]}}

\newcommand{\Specht}[1]{\mathrm{Sp}_{#1}}
\newcommand{\Weyl}[2]{\mathrm{V}_{#1}^{#2}}

\newcommand{\sch}{\Phi}

\newcommand{\prm}{\mathrm{pm}}

\theoremstyle{definition}
\newtheorem*{nameddefinition}{\theoremname}

\begin{document}
\title{Efficient quantum tomography}

\author{Ryan O'Donnell$^*$
\and
John Wright\thanks{Department of Computer Science, Carnegie Mellon University.  Supported by NSF grants CCF-0747250 and CCF-1116594.  The second-named author is also supported by a Simons Fellowship in Theoretical Computer Science.
\texttt{\{odonnell,jswright\}@cs.cmu.edu}}}

\maketitle

\begin{abstract}
    In the quantum state tomography problem, one wishes to estimate an unknown $d$-dimensional mixed quantum state~$\rho$, given few copies.  We show that $O(d/\eps)$ copies suffice to obtain an estimate $\hat{\rho}$ that satisfies $\|\hat{\rho} - \rho\|_F^2 \leq \eps$ (with high probability).  An immediate consequence is that $O(\rank(\rho) \cdot d/\eps^2) \leq O(d^2/\eps^2)$ copies suffice to obtain an $\eps$-accurate estimate in the standard trace distance.  This improves on the best known prior result of~$O(d^3/\eps^2)$ copies for full tomography, and even on the best known prior result of $O(d^2\log(d/\eps)/\eps^2)$ copies for spectrum estimation.  Our result is the first to show that nontrivial tomography can be obtained using a number of copies that is just \emph{linear} in the dimension.

    Next, we generalize these results to show that one can perform efficient principal component analysis on~$\rho$.
    Our main result is that $O(k d/\eps^2)$ copies suffice to output a rank-$k$ approximation $\hat{\rho}$ whose trace distance error is at most $\eps$ more than that of the best rank-$k$ approximator to $\rho$.  This subsumes our above trace distance tomography result and generalizes it to the case when~$\rho$ is not guaranteed to be of low rank.
    A key part of the proof is the analogous generalization of our spectrum-learning results: we show that the largest $k$ eigenvalues of~$\rho$ can be estimated to trace-distance error~$\eps$ using $O(k^2/\eps^2)$ copies.  In turn, this result  relies on a new coupling theorem concerning the Robinson--Schensted--Knuth algorithm that should be of independent combinatorial interest.
\end{abstract}

\section{Introduction}

Quantum state tomography refers to the task of estimating an unknown $d$-dimensional quantum mixed quantum state,~$\rho$, given the ability to prepare and measure~$n$ copies, $\rho^{\otimes n}$.  It is of enormous practical importance for experimental detection of entanglement and the verification of quantum technologies.  For an anthology of recent advances in the area, the reader may consult~\cite{BCG13}.  As stated in its introduction,
\begin{quote}
    \emph{The bottleneck limiting further progress in estimating the states of [quantum] systems has shifted from physical controllability to the problem of handling\dots the exponential scaling of the number of parameters describing quantum many-body states.}
\end{quote}
Indeed, a system consisting of $b$ qubits has dimension $d = 2^b$ and is described by a density matrix with $d^2 = 4^b$ complex parameters.  For practical experiments with, say, $b \leq 10$, it is imperative to use tomographic methods in which~$n$ grows as slowly as possible with~$d$.  For $20$ years or so, the best known method
used $n = O(d^4)$ copies to estimate~$\rho$ to constant error; just recently this was improved~\cite{KRT14} to $n = O(d^3)$.  Despite the practical importance and mathematical elegance of the quantum tomography problem, the optimal dependence of~$n$ on~$d$ remained ``shockingly unknown''~\cite{Har15} as of early 2015.

In this work we analyze known measurements arising from the representation theory of the symmetric and general linear groups $\symm{n}$ and $\gl{d} = \gl{d}(\C)$ --- specifically, the ``Empirical Young Diagram (EYD)'' measurement considered by~\cite{ARS88,KW01}, followed by Keyl's~\cite{KW01,Key06} state estimation measurement based on projection to highest weight vectors.  The former produces a random height-$d$ partition $\blambda \vdash n$ according to the \emph{Schur--Weyl distribution} $\SW{n}{\alpha}$, which depends only on the spectrum $\alpha_1 \geq \alpha_2 \geq \cdots \geq \alpha_d$ of~$\rho$; the latter produces a random $d$-dimensional unitary~$\bU$ according to what may be termed the \emph{Keyl distribution} $\keyl{\lambda}{\rho}$.  Writing $\ullam$ for $(\lambda_1/n, \dots, \lambda_d/n)$, we show the following results:

\begin{theorem}                                     \label{thm:EYD-error-main}
    $\displaystyle \E_{\blambda \sim \SW{n}{\alpha}} \|\ulblam - \alpha\|_2^2 \leq \frac{d}{n}$.
\end{theorem}
\begin{theorem}                                     \label{thm:tomography-error-main}
    $\displaystyle \E_{\substack{\blambda \sim \SW{n}{\alpha} \\ \bU \sim \keyl{\blambda}{\rho}}} \|\bU \diag(\ulblam) \bU^\dagger - \rho\|_F^2 \leq \frac{4d-3}{n}$.
\end{theorem}
In particular, up to a small constant factor, full tomography is no more expensive than spectrum estimation.  These theorems have the following straightforward consequences:
\begin{corollary}                                       \label{cor:EYD-main}
    The spectrum of an unknown rank-$r$ mixed state $\rho \in \C^{d \times d}$ can be estimated to error~$\eps$ in $\ell_2$-distance using $n = O(r / \eps^2)$ copies, or to error~$\eps$ in total variation distance using $n = O(r^2 / \eps^2)$ copies.
\end{corollary}
\begin{corollary}                                       \label{cor:tomography-main}
    An unknown rank-$r$ mixed state $\rho \in \C^{d \times d}$ may be estimated to error~$\eps$ in Frobenius distance using $n = O(d/\eps^2)$ copies, or to error~$\eps$ in trace distance using $n = O(rd/\eps^2)$ copies.
\end{corollary}
\noindent  (These bounds are with high probability; confidence $1-\delta$ may be obtained by increasing the copies by a factor of $\log(1/\delta)$.)\\

The previous best result for spectrum estimation~\cite{HM02,CM06} used $O(r^2 \log(r/\eps)/\eps)$ copies for an $\eps$-accurate estimation in KL-divergence, and hence $O(r^2 \log(r/\eps)/\eps^2)$ copies for an $\eps$-accurate estimation in total variation distance.  The previous best result for tomography is the very recent~\cite[Theorem~2]{KRT14}, which uses $n = O(r d /\eps^2)$ for an $\eps$-accurate estimation in Frobenius distance, and hence $n = O(r^2 d / \eps^2)$ for trace distance.

As for lower bounds, it follows immediately from~\cite[Lemma~5]{FGLE12} and Holevo's bound that $\wt{\Omega}(r d)$ copies are necessary for tomography with trace-distance error~$\eps_0$, where $\eps_0$ is a universal constant.  (Here and throughout $\wt{\Omega}(\cdot)$ hides a factor of $\log d$.) Also, Holevo's bound combined with the existence of $2^{\Omega(d)}$ almost-orthogonal pure states shows that $\wt{\Omega}(d)$ copies are necessary for tomography with Frobenius error~$\eps_0$, even in the rank-$1$ case.  Thus our tomography bounds are optimal up to at most an $O(\log d)$ factor when $\eps$ is a constant.  (Conversely, for constant~$d$, it is easy to show that $\Omega(1/\eps^2)$ copies are necessary even just for spectrum estimation.)  Finally, we remark that $\wt{\Omega}(d^2)$ is a lower bound for tomography with Frobenius error $\eps = \Theta(1/\sqrt{d})$; this also matches our $O(d/\eps^2)$ upper bound.  This last lower bound follows from Holevo and the existence~\cite{Sza82} of $2^{\Omega(d^2)}$ normalized rank-$d/2$ projectors with pairwise Frobenius distance at least $\Omega(1/\sqrt{d})$.

\subsection{Principal component analysis}

Our next results concern principal component analysis (PCA), in which the goal is to find the best rank-$k$ approximator to a mixed state~$\rho \in \C^{d \times d}$, given $1 \leq k \leq d$.
Our algorithm is identical to the Keyl measurement from above, except rather than outputting $\bU \diag(\ulblam) \bU^\dagger$, it outputs  $\bU \diag^{(k)}(\ulblam) \bU^\dagger$ instead, where $\diag^{(k)}(\ulblam)$ means $\diag(\ulblam_1, \dots, \ulblam_k, 0, \dots, 0)$. Writing $\alpha_1 \geq \alpha_2 \geq \ldots \geq \alpha_d$ for the spectrum of~$\rho$, our main result is:

\begin{theorem}                                     \label{thm:pca-error-main}
    $\displaystyle \E_{\substack{\blambda \sim \SW{n}{\alpha} \\ \bU \sim \keyl{\blambda}{\rho}}} \|\bU \diag^{(k)}(\ulblam) \bU^\dagger - \rho\|_1
    	\leq \alpha_{k+1} + \ldots + \alpha_d + 6\sqrt{\frac{k d}{n}}$.
\end{theorem}
\noindent
As the best rank-$k$ approximator to~$\rho$ has trace-distance error $\alpha_{k+1} + \ldots + \alpha_d$, we may immediately conclude:
\begin{corollary}                                       \label{cor:pca-main}
    Using $n = O(kd/\eps^2)$ copies of an unknown mixed state $\rho \in \C^{d \times d}$, one may find a rank-$k$ mixed state $\hat{\rho}$ such that the trace distance of $\hat{\rho}$ from $\rho$ is at most $\eps$ more than that of the optimal rank-$k$ approximator.
\end{corollary}
Since $\alpha_{k+1} = \ldots = \alpha_d = 0$ when $\rho$ has rank~$k$, Corollary~\ref{cor:pca-main} strictly generalizes the trace-distance tomography result from Corollary~\ref{cor:tomography-main}.  We also remark that one could consider performing Frobenius-norm PCA on~$\rho$, but it turns out that this is unlikely to give any improvement in copy complexity over full tomography; see Section~\ref{sec:pca} for details.

As a key component of our PCA result, we investigate the problem of estimating just the largest~$k$ eigenvalues, $\alpha_1, \dots, \alpha_k$, of $\rho$.  The goal here is to use a number of copies depending only on~$k$ and not on $d$ or $\rank(\rho)$.  We show that the standard EYD algorithm achieves this:
\begin{theorem} \label{thm:top-k-eigs-main}
    $\displaystyle \E_{\blambda \sim \SW{n}{\alpha}} \dtvk{k}{\underline{\blambda}}{\alpha} \leq \frac{1.92\,k + .5}{\sqrt{n}}$, where $\dtvk{k}{\beta}{\alpha}$ denotes $\frac12\sum_{i=1}^k |\beta_i - \alpha_i|$.
\end{theorem}
From this we immediately get the following strict generalization of (the total variation distance result in) Corollary~\ref{cor:EYD-main}:
\begin{corollary}                                       \label{cor:truncated-EYD-main}
    The largest $k$ eigenvalues of an unknown mixed state $\rho \in \C^{d \times d}$ can be estimated to error~$\eps$ in in total variation distance using $n = O(k^2 / \eps^2)$ copies.
\end{corollary}
The fact that this result has no dependence on the ambient dimension~$d$ or the rank of~$\rho$ may make it particularly interesting in practice.

\subsection{A coupling result concerning the RSK algorithm}
For our proof of Theorem~\ref{thm:top-k-eigs-main}, we will need to establish a new combinatorial result concerning the Robinson--Schensted--Knuth (RSK) algorithm applied to random words. We assume here the reader is familiar with the RSK correspondence; see Section~\ref{sec:prelims} for a few basics and, e.g.,~\cite{Ful97} for a comprehensive treatment.
\begin{notation}
    Let $\alpha$ be a probability distribution on~$[d] = \{1, 2, \dots, d\}$, and let $\bw \in [d]^n$ be a random word formed by drawing each letter $\bw_i$ independently according to~$\alpha$.  Let $\blambda$ be the shape of the Young tableaus obtained by applying the RSK correspondence to~$\bw$.  We write $\SW{n}{\alpha}$ for the resulting probability distribution on~$\blambda$.
\end{notation}
\begin{notation}
    For $x, y \in \R^d$, we say $x$ \emph{majorizes} $y$, denoted $x \succ y$, if  $\sum_{i=1}^k x_{[i]} \geq \sum_{i=1}^k y_{[i]}$ for all $k \in [d] = \{1,2, \dots, d\}$, with equality for $k = d$.  Here the notation $x_{[i]}$ means the $i$th largest value among $x_1, \dots, x_d$.  We also use the traditional notation $\lambda \unrhd \mu$ instead when $\lambda$ and $\mu$ are partitions of~$n$ (Young diagrams).
\end{notation}

In Section~\ref{sec:coupling} we prove the following theorem.   The proof is entirely combinatorial, and can be read independently of the quantum content in the rest of the paper.
\begin{theorem}                                     \label{thm:coupling-main}
    Let $\alpha$, $\beta$ be probability distributions on $[d]$ with $\beta \succ \alpha$. Then for any $n \in \N$ there is a coupling $(\blambda, \bmu)$ of $\SWdist{n}{\alpha}$ and $\SWdist{n}{\beta}$ such that $\bmu \unrhd \blambda$ always.
\end{theorem}

\subsection{Independent and simultaneous work.}
Independently and simultaneously of our work, Haah et~al.~\cite{HHJ+15} have given a slightly different measurement that also achieves Corollary~\ref{cor:tomography-main}, up to a log factor.  More precisely, their measurement achieves error~$\eps$ in infidelity with $n = O(rd/\eps)\cdot \log(d/\eps)$ copies, or error~$\eps$ in trace distance with $n = O(rd/\eps^2)\cdot \log(d/\eps)$ copies.  They also give a lower bound of $n \geq \Omega(rd/\eps^2)/\log(d/r\eps)$ for quantum tomography with trace distance error~$\eps$.  After seeing a draft of their work, we observed that their measurement can also be shown to achieve expected squared-Frobenius error $\frac{4d-3}{n}$, using the techniques in this paper; the brief details appear at~\cite{Wri15}.

\subsection{Acknowledgments.} We thank Jeongwan Haah and Aram Harrow (and by transitivity, Vlad Voroninski) for bringing~\cite{KRT14} to our attention. We also thank Aram Harrow for pointing us to~\cite{Key06}.  The second-named author would also like to thank  Akshay Krishnamurthy and Ashley Montanaro for helpful discussions.

\section{Preliminaries}\label{sec:prelims}

    We write $\lambda \vdash n$ to denote that $\lambda$ is a \emph{partition} of~$n$; i.e., $\lambda$ is a finite sequence of integers $\lambda_1 \geq \lambda_2 \geq \lambda_3 \geq \cdots$ summing to~$n$. We also say that the \emph{size} of $\lambda$ is $|\lambda| = n$. The \emph{length} (or \emph{height}) of $\lambda$, denoted $\ell(\lambda)$, is the largest~$d$ such that $\lambda_d \neq 0$.  We identify partitions that only differ by trailing zeroes.  A \emph{Young diagram} of \emph{shape}~$\lambda$ is a left-justified set of boxes arranged in rows, with $\lambda_i$ boxes in the $i$th row from the top.  We write $\mu \nearrow \lambda$ to denote that $\lambda$ can be formed from~$\mu$ by the addition of a single box to some row.  A \emph{standard Young tableau}~$T$ of \emph{shape}~$\lambda$ is a filling of the boxes of~$\lambda$ with $[n]$ such that the rows and columns are strictly increasing.  We write $\lambda = \sh(T)$.  Note that $T$ can also be identified with a chain $\emptyset =\lambda^{(0)} \nearrow \lambda^{(1)} \nearrow \cdots \nearrow \lambda^{(n)} = \lambda$, where $\lambda^{(t)}$ is the shape of the Young tableau formed from~$T$ by entries $1 \dd t$.  A \emph{semistandard Young tableau} of shape~$\lambda$ and alphabet~$\calA$ is a filling of the boxes with letters from~$\calA$ such that rows are increasing and columns are strictly increasing.  Here an \emph{alphabet} means a totally ordered set of ``letters'', usually $[d]$.

The quantum measurements we analyze involve the \emph{Schur--Weyl duality theorem}.  The symmetric group $\symm{n}$ acts on $(\C^d)^{\otimes n}$ by permuting factors, and the general linear group $\gl{d}$  acts on it diagonally; furthermore, these actions commute. Schur--Weyl duality states that as an $\symm{n} \times \gl{d}$ representation, we have the following unitary equivalence:
\[
    (\C^d)^{\otimes n} \cong \bigoplus_{\substack{\lambda \vdash n \\ \ell(\lambda) \leq d}} \Specht{\lambda} \otimes \Weyl{\lambda}{d}.
\]
Here we are using  the following notation:    The \emph{Specht modules} $\Specht{\lambda}$ are the irreducible representations spaces of $\symm{n}$, indexed by partitions $\lambda \vdash n$.  We will use the abbreviation $\symdim{\lambda}$ for $\dim(\Specht{\lambda})$; recall this equals the number of standard Young tableaus of shape~$\lambda$.  The \emph{Schur (Weyl) modules} $\Weyl{\lambda}{d}$ are the irreducible polynomial representation spaces of~$\gl{d}$, indexed by partitions (highest weights) $\lambda$ of length at most~$d$. (For more background see, e.g.,~\cite{Har05}.) We will write $\glrep{\lambda} \co \gl{d} \to \mathrm{End}(\Weyl{\lambda}{d})$ for the (unitary) representation itself; the domain of~$\glrep{\lambda}$ naturally extends to all of $\C^{d \times d}$ by continuity. We also write $\ket{T_\lambda}$ for the highest weight vector in~$\Weyl{\lambda}{d}$; it is characterized by the property that $\glrep{\lambda}(A) \ket{T_\lambda} = (\prod_{k=1}^d A_{kk}^{\lambda_i})\ket{T_\lambda}$ if $A = (A_{ij})$ is upper-triangular.

    The character of $\Weyl{\lambda}{d}$ is the \emph{Schur polynomial} $s_\lambda(x_1, \dots, x_d)$, a symmetric, degree-$|\lambda|$, homogeneous polynomial in $x = (x_1, \dots, x_d)$ defined by $s_\lambda(x) = a_{\lambda+\delta}(x)/a_\delta(x)$, where $\delta = (d-1, d-2, \dots, 1, 0)$ and $a_\mu(x) = \det(x_i^{\mu_j})$.  Alternatively, it may be defined as $\sum_{T} \prod_{i=1}^d x_i^{\#_iT}$, where $T$ ranges over all semistandard tableau of shape~$\lambda$ and alphabet~$[d]$, and $\#_i T$ denotes the number of occurrences of~$i$ in~$T$.  We have $\dim(\Weyl{\lambda}{d}) = s_\lambda(1, \dots, 1)$, the number of semistandard Young tableaus in the sum. We'll write $\sch_\lambda(x)$ for the \emph{normalized Schur polynomial} $s_\lambda(x_1, \dots, x_d)/s_\lambda(1, \dots, 1)$.  Finally, we recall the following two formulas, the first following from Stanley's hook-content formula and the Frame--Robinson--Thrall hook-length formula, the second being the Weyl dimension formula:
    \begin{equation}            \label{eqn:ssyt-formula}
        s_\lambda(1, \dots, 1) = 
        \frac{\symdim{\lambda}}{|\lambda|!} \prod_{(i,j) \in \lambda} (d+j-i) = \prod_{1 \leq i < j \leq d} \frac{(\lambda_i - \lambda_j) + (j-i)}{j-i} .
    \end{equation}

    Given a positive semidefinite matrix $\rho \in \C^d$, we typically write $\alpha \in \R^d$ for its \emph{sorted spectrum}; i.e., its eigenvalues $\alpha_1 \geq \alpha_2 \geq \cdots \geq \alpha_d \geq 0$.  When $\rho$ has trace~$1$ it is called a \emph{density matrix} (or \emph{mixed state}), and in this case $\alpha$ defines a (sorted) probability distribution on~$[d]$.  

    We will several times use the following elementary majorization inequality:
    \begin{equation} \label{eqn:majorization-ineq}
        \text{If } c, x, y \in \R^d \text{ are sorted (decreasing) and } x \succ y \text{ then } c \cdot x \geq c \cdot y.
    \end{equation}

    Recall~\cite{Ful97} that the Robinson--Schensted--Knuth correspondence is a certain bijection between strings $w \in \calA^n$ and pairs $(P,Q)$, where $P$ is a semistandard \emph{insertion tableau} filled by the multiset of letters in~$w$, and $Q$ is a standard \emph{recording tableau}, satisfying $\sh(Q) = \sh(P)$.  We write $\RSK(w) = (P,Q)$ and write $\shRSK(w)$ for the common shape of~$P$ and~$Q$, a partition of~$n$ of length at most~$|\calA|$.  One way to characterize $\lambda = \shRSK(w)$ is by \emph{Greene's Theorem}~\cite{Gre74}:  $\lambda_1 + \cdots + \lambda_k$ is the length of the longest disjoint union of~$k$ increasing subsequences in~$w$.  In particular, $\lambda_1 = \LIS(w)$, the length of the longest increasing (i.e., nondecreasing) subsequence in~$w$.  We remind the reader here of the distinction between a \emph{subsequence} of a string, in which the letters need not be consecutive, and a \emph{substring}, in which they are.  We use the notation $w[i \dd j]$ for the substring $(w_i, w_{i+1}, \dots, w_{j}) \in \calA^{j-i+1}$.

    Let $\alpha = (\alpha_1, \dots, \alpha_d)$ denote a probability distribution on alphabet~$[d]$, let $\alpha^{\otimes n}$ denote the associated product probability distribution on~$[d]^n$, and write $\alpha^{\otimes \infty}$ for the product probability distribution on infinite sequences.  We define the associated \emph{Schur--Weyl growth process} to be the (random) sequence
    \begin{equation}			\label{eqn:growth}
        \emptyset = \blambda^{(0)} \nearrow \blambda^{(1)} \nearrow \blambda^{(2)} \nearrow \blambda^{(3)} \nearrow \cdots
    \end{equation}
    where $\bw \sim \alpha^{\otimes \infty}$ and $\blambda^{(t)} = \shRSK(\bw[1 \dd t])$. Note that the marginal distribution on~$\blambda^{(n)}$ is what we call $\SW{n}{\alpha}$.  The Schur--Weyl growth process was studied in, e.g.,~\cite{OC03}, wherein it was noted that the RSK correspondence implies
    \begin{equation}                    \label{eqn:oconnell}
        \Pr[\blambda^{(t)} = \lambda^{(t)}\quad\forall t \leq n] = s_{\lambda^{(n)}}(\alpha)
    \end{equation}
    for any chain $\emptyset = \lambda^{(0)} \nearrow \cdots \nearrow \lambda^{(n)}$. (Together with the fact that $s_\lambda(\alpha)$ is homogeneous of degree~$|\lambda|$, this gives yet another alternate definition of the Schur polynomials.) One consequence of this is that for any $i \in [d]$ we have
    \begin{equation}    \label{eqn:pieri-interp}
        \Pr[\blambda^{(n+1)} = \lambda + e_i \mid \blambda^{(n)} = \lambda] = \frac{s_{\lambda+e_i}(\alpha)}{s_{\lambda}(\alpha)}.
    \end{equation}
    (This formula is correct even when $\lambda + e_i$ is not a valid partition of $n+1$; in this case $s_{\lambda + e_i} \equiv 0$ formally under the determinantal definition.) The above equation is also a probabilistic interpretation of the following special case of \emph{Pieri's rule}:
    \begin{equation}    \label{eqn:pieri}
        (x_1 + \cdots + x_d) s_\lambda(x_1, \dots, x_d) = \sum_{i=1}^d s_{\lambda+e_i}(x_1, \dots, x_d).
    \end{equation}
    We will need the following consequence of~\eqref{eqn:pieri-interp}:
    \begin{proposition}\label{prop:random-walk-majorization}
        Let $\lambda \vdash n$ and let $\alpha \in \R^d$ be a sorted probability distribution.  Then
    \begin{equation}        \label{eqn:rw-maj}
        \left(\frac{s_{\lambda+ e_1}(\alpha)}{s_{\lambda}(\alpha)}, \ldots, \frac{s_{\lambda+ e_d}(\alpha)}{s_{\lambda}(\alpha)}\right)
            \succ
        (\alpha_1, \ldots, \alpha_d).
    \end{equation}
    \end{proposition}
    \begin{proof}
        Let $\beta$ be the reversal of $\alpha$ (i.e.\ $\beta_i = \alpha_{d-i+1}$) and let $(\blambda^{(t)})_{t \geq 0}$ be a Schur--Weyl growth process corresponding to~$\beta$.  By~\eqref{eqn:pieri-interp} and the fact that the Schur polynomials are symmetric, we conclude that the vector on the left of~\eqref{eqn:rw-maj} is $(p_1, \dots, p_d)$, where $p_i = \Pr[\blambda^{(n+1)} = \lambda+e_i \mid \blambda^{(n)} = \lambda]$. Now $p_1 + \cdots + p_k$ is the probability, conditioned on $\blambda^{(n)} = \lambda$, that the $(n+1)$th box in the process enters into one of the first~$k$ rows.  But this is indeed at least $\alpha_1 + \cdots + \alpha_k = \beta_{d} + \cdots + \beta_{d-k+1}$, because the latter represents the probability that the $(n+1)$th letter is $d-k+1$ or higher, and such a letter will always be inserted within the first~$k$ rows under RSK.
    \end{proof}
    A further consequence of~\eqref{eqn:oconnell} (perhaps first noted in~\cite{ITW01}) is that for $\lambda \vdash n$,
    \begin{equation}    \label{eqn:sw-probs}
        \Pr_{\blambda \sim \SW{n}{\alpha}}[\blambda = \lambda] = \symdim{\lambda}s_\lambda(\alpha).
    \end{equation}
    At the same time, as noted in~\cite{ARS88} (see also~\cite[Equation~(36)]{Aud06}) it follows from Schur--Weyl duality that if $\rho \in \C^{d \times d}$ is a density matrix with spectrum~$\alpha$ then
    \[
        \tr(\Pi_\lambda \rho^{\otimes n}) = \symdim{\lambda}s_\lambda(\alpha),
    \]
    where $\Pi_\lambda$ denotes the isotypic projection onto $\Specht{\lambda} \otimes \Weyl{\lambda}{d}$.  Thus we have the identity
    \begin{equation} \label{eqn:eyd-probability}
        \tr(\Pi_\lambda \rho^{\otimes n}) = \Pr_{\blambda \sim \SW{n}{\alpha}}[\blambda = \lambda].
    \end{equation}

\section{Spectrum estimation}

Several groups of researchers suggested the following method for estimating the sorted spectrum $\alpha$ of a quantum mixed state $\rho \in \C^{d \times d}$: measure $\rho^{\otimes n}$ according to the isotypic projectors $\{\Pi_\lambda\}_{\lambda \vdash n}$; and, on obtaining $\blambda$, output the estimate $\hat{\alpha} = \ulblam = (\blambda_1/n, \dots, \blambda_d/n)$.  The measurement is sometimes called ``weak Schur sampling''~\cite{CHW07} and we refer to the overall procedure as the ``Empirical Young Diagram (EYD)'' algorithm.  We remark that the algorithm's behavior depends only on the rank~$r$ of $\rho$; it is indifferent to the ambient dimension~$d$.  So while we will analyze the EYD algorithm in terms of $d$, we will present the results in terms of~$r$.

In~\cite{HM02,CM06} it is shown that  $n = O(r^2 \log(r/\eps)/\eps^2)$ suffices for EYD to obtain $\dkl{\ulblam}{\alpha} \leq 2\eps^2$ and hence $\dtv{\ulblam}{\alpha} \leq \eps$ with high probability.  However we give a different analysis.  By equation~\eqref{eqn:eyd-probability}, the expected $\ell_2^2$-error of the EYD algorithm is precisely $\E_{\blambda \sim \SW{n}{\alpha}} \|\ulblam - \alpha\|_2^2$.  Theorem~\ref{thm:EYD-error-main}, which we prove in this section, bounds this quantity by $\frac{r}{n}$.  Thus
\[
    \E \dtv{\ulblam}{\alpha} = \tfrac12 \E \|\ulblam - \alpha\|_1 \leq \tfrac12 \sqrt{r} \E \|\ulblam - \alpha\|_2 \leq \tfrac12 \sqrt{r} \sqrt{\E \|\ulblam - \alpha\|_2^2} \leq \frac{r}{2\sqrt{n}},
\]
which is bounded by $\eps/4$, say, if $n = 4r^2/\eps^2$.  Thus in this case $\Pr[\dtv{\ulblam}{\alpha} > \eps] < 1/4$.  By a standard amplification (repeating the EYD algorithm $O(\log 1/\delta)$ times and outputting the estimate which is within $2\eps$ total variation distance of the most other estimates), we obtain Corollary~\ref{cor:EYD-main}. 
\\

We give two lemmas, and then the proof of Theorem~\ref{thm:EYD-error-main}.
\begin{lemma}\label{lem:second-moment}
    Let $\alpha \in \R^d$ be a probability distribution.  Then
\begin{equation*}
\E_{\blambda \sim \SWdist{n}{\alpha}} \sum_{i=1}^d \blambda_i^2
\leq  \sum_{i=1}^d (n\alpha_i)^2 + dn.
\end{equation*}
\end{lemma}
\begin{proof}
Define the polynomial function
\begin{equation*}
p^*_2(\lambda) = \sum_{i=1}^{\ell(\lambda)} \left((\lambda_i - i + \tfrac12)^2 - (-i+\tfrac12)^2\right).
\end{equation*}
By Proposition~2.34 and equation~(12) of~\cite{OW15},
$
\E_{\blambda \sim \SWdist{n}{\alpha}} [p^*_2(\blambda)] = n(n-1) \cdot \sum_{i=1}^d \alpha_i^2.
$
Hence,
\begin{equation*}
\E \sum_{i=1}^d \blambda_i^2
= \E\left[ p_2^*(\blambda) + \sum_{i=1}^d (2 i-1)\blambda_i\right]
\leq \E p_2^*(\blambda) +  \sum_{i=1}^d (2i-1)(n/d)
\leq n^2\cdot  \sum_{i=1}^d \alpha_i^2 + dn.
\end{equation*}
Here the first inequality used inequality~\eqref{eqn:majorization-ineq} and $\blambda \succ (n/d, \dots, n/d)$.
\end{proof}
\begin{lemma}\label{lem:lambda-majorize}
Let $\blambda \sim \SW{n}{\alpha}$, where $\alpha \in \R^d$ is a sorted probability distribution.  Then $(\E \blambda_1, \dots, \E \blambda_d) \succ (\alpha_1 n, \dots, \alpha_d n)$.
\end{lemma}
\begin{proof}
    Let $\bw \sim \alpha^{\otimes n}$, so $\blambda$ is distributed as $\shRSK(\bw)$.  The proof is completed by linearity of expectation applied to the fact that $(\blambda_1, \dots, \blambda_d) \succ (\#_1 \bw, \dots, \#_d \bw)$ always, where $\#_k \bw$ denotes the number of times letter~$k$ appears in~$\bw$.  In turn this fact holds by Greene's Theorem: we can form $k$ disjoint increasing subsequences in~$\bw$ by taking all its $1$'s, all its $2$'s, \dots, all its $k$'s.
\end{proof}

\begin{proof}[Proof of Theorem~\ref{thm:EYD-error-main}.]
We have
\begin{multline*}
n^2 \cdot \E_{\blambda \sim \SWdist{n}{\alpha}} \Vert \underline{\blambda} - \alpha\Vert_2^2
= \E \sum_{i=1}^d (\blambda_i -  \alpha_i n)^2
= \E \sum_{i=1}^d  (\blambda_i^2 + (\alpha_i n)^2)- 2 \sum_{i=1}^d (\alpha_i n) \cdot  \E \blambda_i  \\
\leq dn +  2 \sum_{i=1}^d  (\alpha_i n)^2 - 2  \sum_{i=1}^d (\alpha_i n) \cdot  \E \blambda_i
\leq dn +  2 \sum_{i=1}^d  (\alpha_i n)^2 - 2  \sum_{i=1}^d (\alpha_i n) \cdot (\alpha_i n) = dn,
\end{multline*}
where the first inequality used Lemma~\ref{lem:second-moment} and the second used Lemma~\ref{lem:lambda-majorize} and inequality~\eqref{eqn:majorization-ineq} (recall that the coefficients $\alpha_i n$ are decreasing).
Dividing  by~$n^2$ completes the proof.
\end{proof}

\section{Quantum state tomography}
In this section we analyze the tomography algorithm proposed by Keyl~\cite{Key06} based on projection to the highest weight vector. Keyl's method, when applied to density matrix $\rho \in \C^{d \times d}$ with sorted spectrum~$\alpha$, begins by performing weak Schur sampling on~$\rho^{\otimes n}$. Supposing the partition thereby obtained from $\SW{n}{\alpha}$ is $\lambda \vdash n$, the state collapses to ${\frac{1}{s_{\lambda}(\alpha) }\uirrep_{\lambda}(\rho) \in \Weyl{\lambda}{d}}$. The main step of Keyl's algorithm is now to perform a normalized POVM within $\Weyl{\lambda}{d}$ whose outcomes are unitary matrices in $\unitary{d}$.  Specifically, his measurement maps a (Borel) subset $F \subseteq \unitary{d}$ to
\[
    M(F) \coloneqq \int_F \glrep{\lambda}(U) \ket{T_\lambda} \bra{T_\lambda} \glrep{\lambda}(U)^\dagger \cdot \dim(\Weyl{\lambda}{d})\,dU,
\]
where $dU$ denotes Haar measure on $\unitary{d}$.  (To see that this is indeed a POVM --- i.e., that $M \coloneqq M(\unitary{d}) = I$ --- first note that the translation invariance of Haar measure implies $\glrep{\lambda}(V) M \glrep{\lambda}(V)^\dagger = M$ for any $V \in \unitary{d}$.  Thinking of $\glrep{\lambda}$ as an irreducible representation of the unitary group, Schur's lemma implies~$M$ must be a scalar matrix.  Taking traces shows $M$ is the identity.)

We write $\keyl{\lambda}{\rho}$ for the probability distribution on $\unitary{d}$ associated to this POVM; its density with respect to the Haar measure is therefore
\begin{equation}
    \trace\left(\glrep{\lambda}(\tfrac{1}{s_{\lambda}(\alpha) } \rho)  \glrep{\lambda}(U) \ket{T_\lambda} \bra{T_\lambda} \glrep{\lambda}(U)^\dagger \cdot \dim(\Weyl{\lambda}{d}) \right)
    = \sch_{\lambda}(\alpha)^{-1} \cdot \bra{T_\lambda} \glrep{\lambda}(U^\dagger \rho U) \ket{T_\lambda}. \label{eqn:keyl-density}
\end{equation}
Supposing the outcome of the measurement is~$U$, Keyl's final estimate for $\rho$ is $\hat{\rho} = U \diag(\ullam)U^\dagger$.  Thus the expected Frobenius-squared error of Keyl's tomography algorithm is precisely
\[
    \E_{\substack{\blambda \sim \SW{n}{\alpha} \\ \bU \sim \keyl{\blambda}{\rho}}} \|\bU \diag(\ulblam) \bU^\dagger - \rho\|_F^2.
\]
 Theorem~\ref{thm:tomography-error-main}, which we prove in this section, bounds the above quantity by $\frac{4d-3}{n}$.  Let us assume now that $\rank(\rho) \leq r$.  Then $\ell(\ulblam) \leq r$ always and hence the estimate $\bU \diag(\ulblam) \bU^\dagger$ will also have rank at most~$r$.  Thus by Cauchy--Schwarz applied to the singular values of $\bU \diag(\ulblam) \bU^\dagger - \rho$,
\[
    \E \dtr{\bU \diag(\ulblam) \bU^\dagger}{\rho} = \tfrac12 \E \|\bU \diag(\ulblam) \bU^\dagger - \rho\|_1 \leq \tfrac12 \sqrt{2r} \E \|\ulblam - \alpha\|_F \leq \sqrt{r/2} \sqrt{\E \|\ulblam - \alpha\|_F^2} \leq \sqrt{\tfrac{O(rd)}{n}},
\]
and Corollary~\ref{cor:tomography-main} follows just as Corollary~\ref{cor:EYD-main} did.

The remainder of this section is devoted to the proof of Theorem~\ref{thm:tomography-error-main}.

\subsection{Integration formulas}
\begin{notation}
    Let $Z \in \C^{d \times d}$ and let $\lambda$ be a partition of length at most~$d$.  The \emph{generalized power function} $\power{\lambda}$ is defined by
    \[
        \power{\lambda}(Z) = \prod_{k=1}^d \prm_k(Z)^{\lambda_k - \lambda_{k+1}},
    \]
    where $\prm_k(Z)$ denotes the $k$th principal minor of~$Z$ (and $\lambda_{d+1} = 0$).
\end{notation}
As noted by Keyl~\cite[equation~(141)]{Key06}, when $Z$ is positive semidefinite we have $\bra{T_\lambda} \glrep{\lambda}(Z) \ket{T_\lambda} = \power{\lambda}(Z)$; this follows by writing $Z = LL^\dagger$ for $L = (L_{ij})$ lower triangular with nonnegative diagonal and using the fact that $\power{\lambda}(Z) = \power{\lambda}(L^\dagger)^2 = \prod_{k=1}^d L_{kk}^{2\lambda_k}$.  Putting this into~\eqref{eqn:keyl-density} we have an alternate definition for the distribution $\keyl{\lambda}{\rho}$:
\begin{equation}        \label{eqn:keyl-formula}
    \E_{\bU \sim \keyl{\lambda}{\rho}} f(\bU) = \sch_{\lambda}(\alpha)^{-1} \E_{\bU \sim \unitary{d}}  \left[f(\bU) \cdot  \power{\lambda}(\bU^\dagger \rho \bU)\right],
\end{equation}
where $\bU \sim \unitary{d}$ denotes that $\bU$ has the Haar measure. For example, taking $f \equiv 1$ yields the identity
\begin{equation}        \label{eqn:spherical}
    \E_{\bU \sim \unitary{d}} \power{\lambda}(\bU^\dagger \rho \bU) = \sch_{\lambda}(\alpha);
\end{equation}
this expresses the fact that the spherical polynomial of weight~$\lambda$ for $\gl{d}/\unitary{d}$ is precisely the normalized Schur polynomial (see, e.g.,~\cite{Far15}). For a further example, taking $f(U) = \power{\mu}(U^\dagger \rho U)$ and using the fact that $\power{\lambda} \cdot \power{\mu} = \power{\lambda + \mu}$, we obtain
\begin{equation} \label{eqn:11-formula}
    \E_{\bU \sim \keyl{\lambda}{\rho}} \power{\mu}(\bU^\dagger \rho \bU) = \frac{\sch_{\lambda+\mu}(\alpha)}{\sch_{\lambda}(\alpha)}; \quad \text{in particular, } \E_{\bU \sim \keyl{\lambda}{\rho}} (\bU^\dagger \rho \bU)_{1,1} = \frac{\sch_{\lambda+e_1}(\alpha)}{\sch_{\lambda}(\alpha)}.
\end{equation}
For our proof of Theorem~\ref{thm:tomography-error-main}, we will need to develop and analyze a more general formula for the expected diagonal entry $\E (\bU^\dagger \rho \bU)_{k,k}$. We begin with some lemmas.
\begin{definition}
    For $\lambda$ a partition and $m$ a positive integer we define the following partition of height (at most)~$m$:
    \[
        \upperp{\lambda}{m} = (\lambda_1 - \lambda_{m+1}, \ldots, \lambda_m - \lambda_{m+1}).
    \]
    We also define the following ``complementary'' partition $\lambda_{[m]}$ satisfying $\lambda = \upperp{\lambda}{m} + \lowerp{\lambda}{m}$:
    \[
        (\lowerp{\lambda}{m})_i= \begin{cases}
                                                   \lambda_{m+1} & i \leq m, \\
                                                   \lambda_{i}       & i \geq m+1.
                                                   \end{cases}
    \]
\end{definition}
\begin{lemma}\label{lem:gonna-apply-this-twice}
    Let $\rho \in \C^{d \times d}$ be a density matrix with spectrum~$\alpha$ and let $\lambda \vdash n$ have height at most~$d$. Let $m \in [d]$ and let $f_m$ be an $m$-variate symmetric polynomial. Then
    \[
        \E_{\bU \sim \keyl{\lambda}{\rho}} f_m(\bbeta) =
            \sch_{\lambda}(\alpha)^{-1} \cdot
                \E_{\bU \sim \unitary{d}} \left[ f_m(\bbeta) \cdot \sch_{\upperp{\lambda}{m}}(\bbeta) \cdot \power{\lowerp{\lambda}{m}}(\bU^\dagger \rho \bU) \right],
    \]
    where we write $\bbeta = \spec_m(\bU^\dagger \rho \bU)$ for the spectrum of the top-left $m \times m$ submatrix of $\bU^\dagger \rho \bU$.
\end{lemma}
\begin{proof}
    Let $\bV \sim \unitary{m}$ and write $\ol{\bV} = \bV \oplus I$, where $I$ is the $(d-m)$-dimensional identity matrix.  By translation-invariance of Haar measure we have $\bU \ol{\bV} \sim \unitary{d}$, and hence from~\eqref{eqn:keyl-formula},
    \begin{equation}    \label{eqn:halfway}
        \E_{\bU \sim \keyl{\lambda}{\rho}} f_m(\bbeta) =  \sch_{\lambda}(\alpha)^{-1} \E_{\bU\sim \unitary{d}, \bV \sim \unitary{m}}  \left[f_m(\spec_m(\ol{\bV}^\dagger\bU^\dagger \rho \bU \ol{\bV})) \cdot  \power{\lambda}(\ol{\bV}^\dagger\bU^\dagger \rho \bU \ol{\bV})\right].
    \end{equation}
    Note that conjugating a matrix by $\ol{\bV}$ does not change the spectrum of its upper-left $k \times k$ block for any $k \geq m$.
   Thus $\spec_m(\ol{\bV}^\dagger\bU^\dagger \rho \bU \ol{\bV})$ is identical to~$\beta$, and $\prm_{k}(\ol{\bV}^\dagger\bU^\dagger \rho \bU \ol{\bV}) = \prm_k(\bU^\dagger \rho \bU)$ for all $k \geq m$. Thus using $\power{\lambda} = \power{\lowerp{\lambda}{m}} \cdot \power{\upperp{\lambda}{m}}$ we have
     \[
        \eqref{eqn:halfway} = \sch_{\lambda}(\alpha)^{-1} \E_{\bU \sim \unitary{d}}  \left[f_m(\bbeta) \cdot  \power{\lowerp{\lambda}{m}}(\bU^\dagger \rho \bU) \cdot \E_{\bV \sim \unitary{m}} \left[\power{\upperp{\lambda}{m}}(\ol{\bV}^\dagger\bU^\dagger \rho \bU \ol{\bV})\right]\right].
    \]
    But the inner expectation equals $\sch_{\upperp{\lambda}{m}}(\bbeta)$ by~\eqref{eqn:spherical}, completing the proof.
\end{proof}

\begin{lemma}\label{thm:lemma}
    In the setting of Lemma~\ref{lem:gonna-apply-this-twice},
    \begin{equation}    \label{eqn:avg-equals-avg}
        \E_{\bU \sim \keyl{\lambda}{\rho}} \avg_{i=1}^m \left\{(\bU^\dagger \rho \bU)_{i,i}\right\} =
             \sum_{i=1}^m \frac{s_{\lambda^{[m]}+e_i}(1/m)}{\phantom{{}_{+e_i}}s_{\lambda^{[m]}}(1/m)}
             \cdot \frac{\sch_{\lambda + e_i}(\alpha)}{\sch_{\lambda}(\alpha)},
    \end{equation}
    where $1/m$ abbreviates $1/m, \dots, 1/m$ (repeated $m$ times).
\end{lemma}
\begin{remark}  \label{rem:avg}
    The right-hand side of~\eqref{eqn:avg-equals-avg} is also a weighted average --- of the quantities ${\sch_{\lambda + e_i}(\alpha)}/{\sch_{\lambda}(\alpha)}$ --- by virtue of~\eqref{eqn:pieri-interp}.  The lemma also generalizes~\eqref{eqn:11-formula}, as ${s_{\lambda^{[1]}+e_1}(1)}/{s_{\lambda^{[1]}}(1)}$ is simply~$1$.
\end{remark}
\begin{proof}
    On the left-hand side of~\eqref{eqn:avg-equals-avg} we have $\frac{1}{m}$ times the expected trace of the upper-left $m \times m$ submatrix of $\bU^\dagger \rho \bU$.  So by applying Lemma~\ref{lem:gonna-apply-this-twice} with $f_m(\beta) = \frac{1}{m} (\beta_1 + \cdots + \beta_m)$, it is equal to
   \begin{align*}
        &\phantom{=}\ \;\sch_{\lambda}(\alpha)^{-1} \cdot
                \E_{\bU \sim \unitary{d}} \left[ \frac1m (\bbeta_1 + \cdots + \bbeta_m)\cdot \frac{s_{\upperp{\lambda}{m}}(\bbeta)}{s_{\upperp{\lambda}{m}}(1, \dots, 1)} \cdot \power{\lowerp{\lambda}{m}}(\bU^\dagger \rho \bU) \right] \\
        &= \sch_{\lambda}(\alpha)^{-1} \cdot
                \E_{\bU \sim \unitary{d}} \left[ \frac1m \sum_{i=1}^m \frac{s_{\upperp{\lambda}{m}+e_i}(\bbeta)}{s_{\upperp{\lambda}{m}}(1, \dots, 1)} \cdot \power{\lowerp{\lambda}{m}}(\bU^\dagger \rho \bU) \right] \tag{by Pieri~\eqref{eqn:pieri}} \\
        &= \sch_{\lambda}(\alpha)^{-1} \cdot
                 \sum_{i=1}^m \frac{\;s_{\upperp{\lambda}{m}+e_i}(1, \dots, 1)}{m \cdot s_{\upperp{\lambda}{m}}(1, \dots, 1)} \cdot \E_{\bU \sim \unitary{d}} \left[\sch_{\upperp{\lambda}{m}+e_i}(\bbeta)  \cdot \power{\lowerp{\lambda}{m}}(\bU^\dagger \rho \bU) \right] \\
        &= \sch_{\lambda}(\alpha)^{-1} \cdot
                 \sum_{i=1}^m \frac{\;s_{\upperp{\lambda}{m}+e_i}(1, \dots, 1)}{m \cdot s_{\upperp{\lambda}{m}}(1, \dots, 1)} \cdot \sch_{\lambda+e_i}(\alpha),
   \end{align*}
   where in the last step we used Lemma~\ref{lem:gonna-apply-this-twice} again, with $f_m \equiv 1$ and $\lambda + e_i$ in place of~$\lambda$. But this is equal to the right-hand side of~\eqref{eqn:avg-equals-avg}, using the homogeneity of Schur polynomials.
\end{proof}

\begin{lemma}\label{lem:diagonal}
    Assume the setting of Lemma~\ref{lem:gonna-apply-this-twice}. Then $\eta_i \coloneqq \E_{\bU \sim \keyl{\lambda}{\rho}} (\bU^\dagger \rho \bU)_{m,m}$ is a convex combination of the quantities $R_i \coloneqq {\sch_{\lambda + e_i}(\alpha)}/{\sch_{\lambda}(\alpha)}$, $1 \leq i \leq m$.\footnote{To be careful, we may exclude all those~$i$ for which $\lambda + e_i$ is an invalid partition and thus $R_i = 0$.}
\end{lemma}
\begin{proof}
    This is clear for $m = 1$. For $m > 1$, Remark~\ref{rem:avg} implies
    \[
        \avg_{i=1}^m \{\eta_i\} = p_1 R_1 + \cdots + p_m R_m, \quad
        \avg_{i=1}^{m-1} \{\eta_i\} = q_1 R_1 + \cdots + q_{m} R_m,
    \]
    where $p_1 + \cdots + p_m = q_{1} + \cdots + q_{m} = 1$ and $q_m = 0$. Thus
    $
        \eta_i = \sum_{i=1}^m r_i R_i$, where ${r_i = (mp_i - (m-1)q_i)}$,
    and evidently $\sum_{i=1}^m r_i = m - (m-1) = 1$.  It remains to verify that each $r_i  \geq 0$. This is obvious for $i = m$; for $i < m$, we must check that
    \begin{equation}            \label{eqn:weird-ratio}
         \frac{s_{\upperp{\lambda}{m}+e_i}(1, \dots, 1)}{\phantom{{}_{+e_i}}s_{\upperp{\lambda}{m}}(1, \dots, 1)} \geq \frac{s_{\upperp{\lambda}{m-1}+e_i}(1, \dots, 1)}{\phantom{{}_{+e_i}}s_{\upperp{\lambda}{m-1}}(1, \dots, 1)}.
    \end{equation}
    Using the Weyl dimension formula from~\eqref{eqn:ssyt-formula}, one may explicitly compute that the ratio of the left side of~\eqref{eqn:weird-ratio} to the right side is precisely $1 + \frac{1}{(\lambda_i - \lambda_m) + (m-i)} \geq 1$. This completes the proof.
\end{proof}
We will in fact only need the following corollary:
\begin{corollary}\label{cor:diagonal}
    Let $\rho \in \C^{d \times d}$ be a density matrix with spectrum~$\alpha$ and let $\lambda \vdash n$ have height at most~$d$. Then $\E_{\bU \sim \keyl{\lambda}{\rho}} (\bU^\dagger \rho \bU)_{m,m} \geq {\sch_{\lambda + e_m}(\alpha)}/{\sch_{\lambda}(\alpha)}$ for every $m \in [d]$,
\end{corollary}
\begin{proof}
    This is immediate from Lemma~\ref{lem:diagonal} and the fact that $\sch_{\lambda + e_i}(\alpha) \geq \sch_{\lambda + e_m}(\alpha)$ whenever $i < m$ (assuming $\lambda + e_i$ is a valid partition).  This latter fact was recently proved by Sra~\cite{Sra15}, verifying a conjecture of Cuttler et al.~\cite{CGS11}.
\end{proof}

\subsection{Proof of Theorem~\ref{thm:tomography-error-main}}

Throughout the proof we assume $\blambda \sim \SW{n}{\alpha}$ and $\bU \sim \keyl{\blambda}{\rho}$.  
We have
\begin{multline}
     n^2 \cdot \E_{\blambda, \bU} \|\bU \diag(\ulblam) \bU^\dagger - \rho\|_F^2
  = n^2 \cdot \E_{\blambda, \bU} \|\diag(\ulblam) - \bU^\dagger \rho \bU\|_F^2 \\
= \E_{\blambda} \sum_{i=1}^d \blambda_i^2 + \sum_{i=1}^d (\alpha_i n)^2 -2n \E_{\blambda, \bU} \sum_{i=1}^d \blambda_i (\bU^\dagger \rho \bU)_{i,i}
\leq dn + 2\sum_{i=1}^d (\alpha_i n)^2 - 2n \E_\lambda \sum_{i=1}^d \blambda_i \E_{\bU}(\bU^\dagger \rho \bU)_{i,i},  \label{eqn:tomog1}
\end{multline}
using Lemma~\ref{lem:second-moment}.  Then by Corollary~\ref{cor:diagonal},
\begin{multline}
\E_\lambda \sum_{i=1}^d \blambda_i \E_{\bU}(\bU^\dagger \rho \bU)_{i,i} \geq \E_\lambda \sum_{i=1}^d \blambda_i \frac{\sch_{\blambda+e_i}(\alpha)}{\sch_{\blambda}(\alpha)}
= \E_\lambda \sum_{i=1}^d \blambda_i \frac{s_{\blambda+e_i}(\alpha)}{\phantom{{}_{+e_i}}s_{\blambda}(\alpha)} \frac{\phantom{{}_{+e_i}}s_{\blambda}(1, \dots, 1)}{s_{\blambda+e_i}(1, \dots, 1)} \\
\geq \E_{\blambda} \sum_{i=1}^d \blambda_i \frac{s_{\blambda+e_i}(\alpha)}{\phantom{{}_{+e_i}}s_{\blambda}(\alpha)}\left(2-\frac{s_{\blambda+e_i}(1, \dots, 1)}{\phantom{{}_{+e_i}}s_{\blambda}(1, \dots, 1)}\right) =
2\E_{\blambda} \sum_{i=1}^d \blambda_i \frac{s_{\blambda+e_i}(\alpha)}{\phantom{{}_{+e_i}}s_{\blambda}(\alpha)}- \E_{\blambda} \sum_{i=1}^d \blambda_i \frac{s_{\blambda+e_i}(\alpha)}{\phantom{{}_{+e_i}}s_{\blambda}(\alpha)}\frac{s_{\blambda+e_i}(1, \dots, 1)}{\phantom{{}_{+e_i}}s_{\blambda}(1, \dots, 1)}, \label{eqn:tomog2}
\end{multline}
where we used $r \geq 2-\tfrac{1}{r}$ for $r > 0$. We lower-bound the first term in~\eqref{eqn:tomog2} by first using the inequality~\eqref{eqn:majorization-ineq} and Proposition~\ref{prop:random-walk-majorization}, and then using inequality~\eqref{eqn:majorization-ineq} and Lemma~\ref{lem:lambda-majorize} (as in the proof of Theorem~\ref{thm:EYD-error-main}):
\begin{equation}    \label{eqn:tomog3}
    2\E_{\blambda} \sum_{i=1}^d \blambda_i \frac{s_{\blambda+e_i}(\alpha)}{\phantom{{}_{+e_i}}s_{\blambda}(\alpha)}
    \geq 2\E_{\blambda} \sum_{i=1}^d \blambda_i \alpha_i \geq 2 n\sum_{i=1}^d  \alpha_i^2.
\end{equation}
As for the second term in~\eqref{eqn:tomog2}, we use~\eqref{eqn:sw-probs} and the first formula in~\eqref{eqn:ssyt-formula} to compute
\begin{align}
    \E_{\blambda} \sum_{i=1}^d \blambda_i \frac{s_{\blambda+e_i}(\alpha)}{\phantom{{}_{+e_i}}s_{\blambda}(\alpha)}\frac{s_{\blambda+e_i}(1, \dots, 1)}{\phantom{{}_{+e_i}}s_{\blambda}(1, \dots, 1)}
&= \sum_{i=1}^d
\sum_{\lambda \vdash n} \symdim{\lambda}  s_{\lambda}(\alpha)\cdot \lambda_i \cdot \frac{s_{\lambda+e_i}(\alpha)}{\phantom{{}_{+e_i}}s_{\lambda}(\alpha)} \frac{\symdim{\lambda+e_i}(d+\lambda_i - i+1)}{\symdim{\lambda}(n+1)} \nonumber\\
&= \sum_{i=1}^d
\sum_{\lambda \vdash n} \symdim{\lambda+e_i} s_{\lambda+e_i}(\alpha)\cdot \frac{\lambda_i(d-i+\lambda_i+1)}{n+1}\nonumber\\
&\leq \sum_{i=1}^d \E_{\blambda' \sim \SW{n+1}{\alpha}} \frac{(\blambda'_i-1)(d-i+\blambda'_i)}{n+1} \tag{by $\eqref{eqn:sw-probs}$ again}\\
&\leq \frac{1}{n+1} \left(\E_{\blambda' \sim \SW{n+1}{\alpha}} \sum_{i=1}^d (\blambda'_i)^2 + \E_{\blambda' \sim \SW{n+1}{\alpha}} \sum_{i=1}^d (d-i-1) \blambda'_i\right) \nonumber\\
&\leq \frac{1}{n+1} \left((n+1)n\sum_{i=1}^d \alpha_i^2 + \sum_{i=1}^d (d+i-2)((n+1)/d)\right) \nonumber\\
&=  n\sum_{i=1}^d \alpha_i^2 + \frac32d - \frac32\label{eqn:tomog4}
\end{align}
where the last inequality is deduced exactly as in the proof of Lemma~\ref{lem:second-moment}.  Finally, combining \eqref{eqn:tomog1}--\eqref{eqn:tomog4} we get
\[
    n^2 \cdot \E_{\blambda, \bU} \|\bU \diag(\ulblam) \bU^\dagger - \rho\|_F^2 \leq 4dn - 3n.
\]
Dividing both sides by~$n^2$ completes the proof. \hfill $\square$

\section{Truncated spectrum estimation}

In this section we prove Theorem~\ref{thm:top-k-eigs-main}, from which Corollary~\ref{cor:truncated-EYD-main} follows in the same way as Corollary~\ref{cor:EYD-main}.  The key lemma involved is the following:
\begin{lemma}\label{lem:k-root-n}
Let $\alpha \in \R^d$ be a sorted probability distribution.  Then for any $k \in [d]$,
\begin{equation*}
\E_{\blambda \sim \SWdist{n}{\alpha}}  \sum_{i=1}^k \blambda_i \leq \sum_{i=1}^k \alpha_in + 2\sqrt{2} k\sqrt{n}.
\end{equation*}
\end{lemma}

We remark that it is easy to \emph{lower}-bound this expectation by $\sum_{i=1}^k  \alpha_i n$ via Lemma~\ref{lem:lambda-majorize}.  We now show how to deduce Theorem~\ref{thm:top-k-eigs-main} from Lemma~\ref{lem:k-root-n}.  Then in Section~\ref{sec:proof-of-lemma} we prove the lemma.

\begin{proof}[Proof of Theorem~\ref{thm:top-k-eigs-main}.]
    Let $\bw \sim \alpha^{\otimes n}$, let $\rsk(\bw) = (\bP,\bQ)$, and let $\blambda = \sh(\bP)$, so $\blambda \sim \SW{n}{\alpha}$.  Write $\bw'$ for the string formed from~$\bw$ by deleting all letters bigger than~$k$.  Then it is a basic property of the RSK algorithm that $\rsk(\bw')$ produces the insertion tableau $\bP'$ formed from~$\bP$ by deleting all boxes with labels bigger than~$k$.  Thus $\blambda' = \sh(\bP') = \shRSK(\bw')$.  Denoting $\alpha_{[k]} = \alpha_1 + \cdots + \alpha_k$, we have $\blambda' \sim \SW{\bm}{\alpha'}$, where $\bm \sim \mathrm{Binomial}(n,\alpha_{[k]})$ and $\alpha'$ denotes $\alpha$ conditioned on the first~$k$ letters; i.e., $\alpha' = ({\alpha_i}/{\alpha_{[k]}})_{i=1}^k$.
    Now by the triangle inequality,
    \begin{equation}
        2n \cdot \E \dtvk{k}{\underline{\blambda}}{\alpha}
    = \E \sum_{i=1}^k \left\vert \blambda_i - \alpha_in\right\vert
    \leq \E \sum_{i=1}^k (\blambda_i - \blambda'_i) + \E \sum_{i=1}^k \left\vert \blambda'_i - \alpha'_i\bm\right\vert +  \sum_{i=1}^k \left\vert \alpha'_i \bm - \alpha_i n \right\vert.  \label{eqn:trunc-deduction}
    \end{equation}
    The first quantity in~\eqref{eqn:trunc-deduction} is at most $2\sqrt{2}k\sqrt{n}$, using Lemma~\ref{lem:k-root-n} and the fact that $\E[\sum_{i=1}^k \blambda'_i] = \E[\bm] = \sum_{i=1}^k \alpha_i n$.  The second quantity in~\eqref{eqn:trunc-deduction} is at most $k \sqrt{n}$ using Theorem~\ref{thm:EYD-error-main}:
    \[
         \E \sum_{i=1}^k \left\vert \blambda'_i - \alpha'_i\bm\right\vert = \E_{\bm} \bm \cdot \E_{\blambda'} \|\ul{\blambda'} - \alpha'\|_1 \leq \E_{\bm}\bm \sqrt{k}\sqrt{\E_{\blambda'} \|\ul{\blambda'} - \alpha'\|_2^2} \leq k \E_{\bm} \sqrt{\bm} \leq k\sqrt{n}.
    \]
    And the third quantity in~\eqref{eqn:trunc-deduction} is at most~$\sqrt{n}$:
    \[
        \E_{\bm} \sum_{i=1}^k \left\vert \alpha'_i \bm - \alpha_i n \right\vert = \E_{\bm} \sum_{i=1}^k \tfrac{\alpha_i}{\alpha_{[k]}} \left\vert \bm - \alpha_{[k]} n \right\vert = \E_{\bm} |\bm - \alpha_{[k]} n|\leq \stddev(\bm) \leq \sqrt{n}.
    \]
    Thus $2n \cdot \E \dtvk{k}{\underline{\blambda}}{\alpha} \leq ((2\sqrt{2}+1)k+1)\sqrt{n}$, and dividing by $2n$ completes the proof.
\end{proof}

\subsection{Proof of Lemma~\ref{lem:k-root-n}}          \label{sec:proof-of-lemma}

Our proof of Lemma~\ref{lem:k-root-n} is essentially by reduction to the case when $\alpha$ is the uniform distribution and $k = 1$.  We thus begin by analyzing the uniform distribution.

\subsubsection{The uniform distribution case}
In this subsection we will use the abbreviation $(1/d)$ for the uniform distribution $(1/d, \dots, 1/d)$ on~$[d]$.  Our goal is the following fact, which is of independent interest:

\begin{theorem} \label{thm:k-root-n-unif}
$\displaystyle \E_{\blambda \sim \SW{n}{1/d}}  \blambda_1 \leq n/d + 2\sqrt{n}$.
\end{theorem}
We remark that Theorem~\ref{thm:k-root-n-unif} implies Lemma~\ref{lem:k-root-n} (with a slightly better constant) in the case of $\alpha = (1/d, \dots, 1/d)$, since of course $\blambda_i \leq \blambda_1$ for all $i \in [k]$.  Also, by taking $d \to \infty$ we recover the well known fact that $\E \blambda_1 \leq 2\sqrt{n}$ when $\blambda$ has the Plancherel distribution.  Indeed, our proof of Theorem~\ref{thm:k-root-n-unif} extends the original proof of this fact by Vershik and Kerov~\cite{VK85} (cf.~the exposition in~\cite{Rom14}).
\begin{proof}
    Consider the Schur--Weyl growth process under the uniform distribution $(1/d, \dots, 1/d)$ on~$[d]$.  For $m \geq 1$ we define
     \[
        \delta_m = \E[\blambda_1^{(m)} - \blambda_1^{(m-1)}] = \Pr[\text{the $m$th box enters into the $1$st row}] = \E_{\blambda \sim \SW{m-1}{1/d}} \frac{s_{\blambda+e_1}(1/d)}{s_{\blambda}(1/d)},
     \]
     where we used~\eqref{eqn:pieri-interp}. By Cauchy--Schwarz and identity~\eqref{eqn:sw-probs},
    \begin{align}
        \delta_m^2 &\leq \E_{\blambda \sim \SW{m-1}{1/d}} \left(\frac{s_{\blambda+e_1}(1/d)}{s_{\blambda}(1/d)}\right)^2
    = \sum_{\lambda \vdash m-1} \dim(\lambda) s_{\lambda}(1/d) \cdot \left(\frac{s_{\lambda+e_1}(1/d)}{s_{\lambda}(1/d)}\right)^2
    \nonumber\\
    &= \sum_{\lambda \vdash m-1} \dim(\lambda) s_{\lambda+e_1}(1/d) \cdot \left(\frac{s_{\lambda+e_1}(1/d)}{s_{\lambda}(1/d)}\right) =
\sum_{\lambda \vdash m-1} \dim(\lambda+e_1) s_{\lambda+e_1}(1/d) \cdot \left(\frac{d + \lambda_1}{d m}\right) \label{eqn:switch}\\
&\hspace{7.14cm}\leq \E_{\blambda \sim \SW{m}{1/d}}\left(\frac{d + \blambda_1}{d m}\right)
= \left(\frac{d + \delta_1 + \ldots + \delta_m}{d m}\right), \nonumber
    \end{align}
    where the ratio in~\eqref{eqn:switch} was computed using the first formula of~\eqref{eqn:ssyt-formula} (and the homogeneity of Schur polynomials).  Thus we have established the following recurrence:
    \begin{equation}            \label{eqn:sw-recur}
        \delta_m  \leq \frac{1}{\sqrt{dm}}\sqrt{d + \delta_1 + \cdots + \delta_m}.
    \end{equation}
    We will now show by induction that $\delta_m \leq \frac{1}{d} + \frac{1}{\sqrt{m}}$ for all $m \geq 1$.  Note that this will complete the proof, by summing over $m \in [n]$.  The base case, $m = 1$, is immediate since $\delta_1 = 1$.  For general $m > 1$, think of $\delta_1, \dots, \delta_{m-1}$ as fixed and $\delta_m$ as variable.  Now if~$\delta_m$ satisfies~\eqref{eqn:sw-recur}, it is bounded above by the (positive) solution $\delta^*$ of
    \[
        \delta = \frac{1}{\sqrt{dm}}\sqrt{c+\delta}, \qquad \text{where $c = d + \delta_1 + \cdots + \delta_{m-1}$.}
    \]
    Note that if $\delta > 0$ satisfies
    \begin{equation} \label{eqn:if-delta}
        \delta \geq \frac{1}{\sqrt{dm}}\sqrt{c+\delta}
    \end{equation}
    then it must be that $\delta \geq \delta^* \geq \delta_m$.  Thus it suffices to show that~\eqref{eqn:if-delta} holds for $\delta = \frac{1}{d} + \frac{1}{\sqrt{m}}$.  But indeed,
    \begin{multline*}
           \frac{1}{\sqrt{dm}}\sqrt{c+\frac{1}{d} + \frac{1}{\sqrt{m}}}
        = \frac{1}{\sqrt{dm}}\sqrt{d + \delta_1 + \cdots + \delta_{m-1} +\frac{1}{d} + \frac{1}{\sqrt{m}}}
    \\ \leq \frac{1}{\sqrt{dm}}\sqrt{d + \sum_{i=1}^m \left(\frac1d + \frac{1}{\sqrt{i}}\right)}
        \leq \frac{1}{\sqrt{dm}}\sqrt{d + \frac{m}{d} + 2\sqrt{m}}
        = \frac{1}{\sqrt{dm}}\left(\sqrt{d} + \sqrt{\frac{m}{d}}\right)
        = \frac1d + \frac{1}{\sqrt{m}},
    \end{multline*}
    where the first inequality used induction.  The proof is complete.
\end{proof}

\subsubsection{Reduction to the uniform case}
\begin{proof}[Proof of Lemma~\ref{lem:k-root-n}]
    Given the sorted distribution $\alpha$ on $[d]$, let $\beta$ be the sorted probability distribution on~$[d]$ defined, for an appropriate value of~$m$, as
    \begin{equation*}
    \beta_1 = \alpha_1,  \ldots, \beta_k = \alpha_k,\quad \beta_{k+1} = \ldots = \beta_{m} = \alpha_{k+1} > \beta_{m+1} \geq 0,
    \quad \beta_{m+2} = \ldots = \beta_{d} = 0.
    \end{equation*}
    In other words, $\beta$ agrees with $\alpha$ on the first~$k$ letters and is otherwise uniform, except for possibly a small ``bump'' at $\beta_{m+1}$.  By construction we have $\beta \succ \alpha$. Thus it follows from our coupling result, Theorem~\ref{thm:coupling-main}, that
    \[
        \E_{\blambda \sim \SWdist{n}{\alpha}}  \sum_{i=1}^k \blambda_i
            \leq \E_{\bmu \sim \SWdist{n}{\beta}}  \sum_{i=1}^k \bmu_i,
    \]
    and hence it suffices to prove the lemma for $\beta$ in place of~$\alpha$.  Observe that $\beta$ can be expressed as a mixture
    \begin{equation}        \label{eqn:mixture}
        \beta = p_1 \cdot \calD_1 + p_2 \cdot \calD_2 + p_3 \cdot \calD_3,
    \end{equation}
    of a certain distribution $\calD_1$ supported on~$[k]$, the uniform distribution $\calD_2$ on~$[m]$, and the uniform distribution $\calD_3$ on~$[m+1]$.  We may therefore think of a draw $\bmu \sim \SWdist{n}{\beta}$ occurring as follows.  First, $[n]$ is partitioned into three subsets $\bI_1, \bI_2, \bI_3$ by including each $i \in [n]$ into $\bI_j$ independently with probability~$p_j$. Next we draw strings $\bw^{(j)} \sim \calD_j^{\otimes \bI_j}$ independently for $j \in [3]$.  Finally, we let $\bw = (\bw^{(1)}, \bw^{(2)}, \bw^{(3)}) \in [d]^n$ be the natural composite string and define $\bmu = \shRSK(\bw)$.  Let us also write $\bmu^{(j)} = \shRSK(\bw^{(j)})$ for $j \in [3]$.  We now claim that
    \[
        \sum_{i=1}^k \bmu_i \leq \sum_{i=1}^k \bmu_i^{(1)}  + \sum_{i=1}^k \bmu_i^{(2)}   + \sum_{i=1}^k \bmu_i^{(3)}
    \]
    always holds. Indeed, this follows from Greene's Theorem: the left-hand side is $|\bs|$, where $\bs \in [d]^n$ is a maximum-length disjoint union of~$k$ increasing subsequences in~$\bw$; the projection of $\bs^{(j)}$ onto coordinates~$\bI_j$ is a disjoint union of~$k$ increasing subsequences in $\bw^{(j)}$ and hence the right-hand side is at least $|\bs^{(1)}| + |\bs^{(2)}|  + |\bs^{(3)}| = |\bs|$.  Thus to complete the proof of the lemma, it suffices to show
    \begin{equation}    \label{eqn:complete-me}
        \E \sum_{i=1}^k \bmu_i^{(1)}  + \E \sum_{i=1}^k \bmu_i^{(2)}   + \E \sum_{i=1}^k \bmu_i^{(3)} \leq \sum_{i=1}^k \alpha_i n + 2\sqrt{2}\,k\sqrt{n}.
    \end{equation}
    Since $\calD_1$ is supported on~$[k]$, the first expectation above is equal to $\E[|\bw^{(1)}|] = p_1 n$.  By (the remark just after) Theorem~\ref{thm:k-root-n-unif}, we can bound the second expectation as
    \[
        \E \sum_{i=1}^k \bmu_i^{(2)} \leq k \E \bmu_1^{(2)} \leq k \E|\bw^{(2)}|/m + 2 k\E\sqrt{|\bw^{(2)}|} \leq k(p_2 n)/m + 2k\sqrt{p_2 n}.
    \]
    Similarly the third expectation in~\eqref{eqn:complete-me} is bounded by $k(p_3 n)/(m+1) + 2k\sqrt{p_3 n}$. Using $\sqrt{p_2} + \sqrt{p_3} \leq \sqrt{2}$, we have upper-bounded the left-hand side of~\eqref{eqn:complete-me} by
    \[
         (p_1 + p_2 \tfrac{k}{m} + p_3 \tfrac{k}{m+1}) n + 2\sqrt{2}\,k\sqrt{n} = \left(\sum_{i=1}^k \beta_i\right) n + 2\sqrt{2}\,k\sqrt{n},
    \]
    as required.
\end{proof}

\section{Principal component analysis}\label{sec:pca}

In this section we analyze a straightforward modification to Keyl's tomography algorithm
that allows us to perform principal component analysis on an unknown density matrix $\rho \in \C^{d \times d}$.
The PCA algorithm is the same as Keyl's algorithm, except that having measured $\blambda$ and $\bU$, it outputs the rank-$k$ matrix $\bU \diag^{(k)}(\underline{\blambda}) \bU^\dagger$ rather than the potentially full-rank matrix $\bU \diag(\underline{\blambda}) \bU^\dagger$.
Here we recall the notation $\diag^{(k)}(\underline{\blambda})$ for the $d\times d$ matrix $\diag(\ulblam_1, \dots, \ulblam_k, 0, \dots, 0)$.

Before giving the proof of Theorem~\ref{thm:pca-error-main}, let us show why the case of Frobenius-norm PCA appears to be less interesting than the case of trace-distance PCA.
The goal for Frobenius PCA would be to output a rank-$k$ matrix $\widetilde{\rho}$ satisfying
\begin{equation*}
	\|\widetilde{\rho}- \rho\|_F
	\leq \sqrt{\alpha_{k+1}^2 + \ldots + \alpha_{d}^2} + \eps,
\end{equation*}
with high probability, while trying to minimize the number of copies~$n$ as a function of $k$, $d$, and $\eps$.
However, even when $\rho$ is guaranteed to be of rank~$1$, it is likely that any algorithm will require $n =\Omega(d/\eps^2)$ copies to output an $\eps$-accurate rank-$1$ approximator~$\widetilde{\rho}$.
This is because such an approximator will satisfy $\|\widetilde{\rho}-\rho\|_1 \leq \sqrt{2}\cdot \|\widetilde{\rho}-\rho\|_F = O(\eps)$,
and it is likely that $n = \Omega(d/\eps^2)$  copies of~$\rho$ are required for such a guarantee (see, for example, the lower bounds of~\cite{HHJ+15},
which show that $n = \Omega(\tfrac{d}{\eps^2 \log(d/\eps)})$ copies are necessary for tomography of rank-$1$ states.).
Thus, even in the simplest case of rank-$1$ PCA of rank-$1$ states, we probably cannot improve on the $n = O(d/\eps^2)$ copy complexity for full tomography given by Corollary~\ref{cor:tomography-main}.

Now we prove Theorem~\ref{thm:pca-error-main}.  We note that the proof shares many of its steps with the proof of Theorem~\ref{thm:tomography-error-main}.
\begin{proof}[Proof of Theorem~\ref{thm:pca-error-main}]
Throughout the proof we assume $\blambda \sim \SW{n}{\alpha}$ and $\bU \sim \keyl{\blambda}{\rho}$.
We write $\bR$ for the lower-right $(d-k) \times (d-k)$ submatrix of $\bU^\dagger \rho \bU$
and we write $\bGamma = \bU^\dagger \rho \bU - \bR$.  Then
\begin{equation}
\E_{\blambda, \bU} \| \bU \diag^{(k)}(\ulblam) \bU^\dagger -\rho\|_1
= \E_{\blambda, \bU} \|\diag^{(k)}(\ulblam) - \bU^\dagger \rho \bU\|_1
\leq \E_{\blambda, \bU} \|\diag^{(k)}(\ulblam) - \bGamma\|_1 + \E_{\blambda, \bU} \| \bR \|_1.\label{eqn:pca1}
\end{equation}
We can upper-bound the first term in~\eqref{eqn:pca1} using
\begin{equation}\label{eqn:never-gonna-reference-this-prelude}
\E_{\blambda, \bU} \|\diag^{(k)}(\ulblam) - \bGamma\|_1
\leq \sqrt{2k} \E_{\blambda, \bU} \|\diag^{(k)}(\ulblam) - \bGamma\|_F
\leq \sqrt{2k} \E_{\blambda, \bU} \|\diag(\ulblam) - \bU^\dagger \rho \bU\|_F
\leq \sqrt{\frac{8kd}{n}}.
\end{equation}
The first inequality is Cauchy--Schwarz together with the fact that $\rank(\diag^{(k)}(\ulblam) - \bGamma) \leq 2k$ (since the matrix is nonzero only in its first~$k$ rows and columns).  The second inequality uses that $\diag(\ulblam) - \bU^\dagger \rho \bU$ is formed from $\diag^{(k)}(\ulblam) - \bGamma$ by adding a matrix, $\diag(\ulblam)-\diag^{(k)}(\ulblam)-\bR$, of disjoint support; this can only increase the squared Frobenius norm (sum of squares of entries). Finally, the third inequality uses Theorem~\ref{thm:tomography-error-main}.
To analyze the second term in~\eqref{eqn:pca1}, we note that $\bR$ is a principal submatrix of $\bU^\dagger \rho \bU$, and so it is positive semidefinite.
As a result,
\begin{equation}\label{eqn:never-gonna-reference-this}
\E_{\blambda, \bU} \| \bR \|_1
= \E_{\blambda, \bU} \tr(\bR)
= 1- \E_{\blambda, \bU} \tr(\bGamma).
\end{equation}
By Corollary~\ref{cor:diagonal},
\begin{multline}
\E_{\blambda, \bU} \tr(\bGamma)
= \E_{\blambda} \sum_{i=1}^k \E_{\bU}(\bU^\dagger \rho \bU)_{i,i} \geq \E_{\blambda} \sum_{i=1}^k \frac{\sch_{\blambda+e_i}(\alpha)}{\sch_{\blambda}(\alpha)}
= \E_{\blambda} \sum_{i=1}^k \frac{s_{\blambda+e_i}(\alpha)}{\phantom{{}_{+e_i}}s_{\blambda}(\alpha)} \frac{\phantom{{}_{+e_i}}s_{\blambda}(1, \dots, 1)}{s_{\blambda+e_i}(1, \dots, 1)} \\
\geq \E_{\blambda} \sum_{i=1}^k \frac{s_{\blambda+e_i}(\alpha)}{\phantom{{}_{+e_i}}s_{\blambda}(\alpha)}\left(2-\frac{s_{\blambda+e_i}(1, \dots, 1)}{\phantom{{}_{+e_i}}s_{\blambda}(1, \dots, 1)}\right) =
2\E_{\blambda} \sum_{i=1}^k  \frac{s_{\blambda+e_i}(\alpha)}{\phantom{{}_{+e_i}}s_{\blambda}(\alpha)}- \E_{\blambda} \sum_{i=1}^k \frac{s_{\blambda+e_i}(\alpha)}{\phantom{{}_{+e_i}}s_{\blambda}(\alpha)}\frac{s_{\blambda+e_i}(1, \dots, 1)}{\phantom{{}_{+e_i}}s_{\blambda}(1, \dots, 1)}, \label{eqn:pca2}
\end{multline}
where we used $r \geq 2-\tfrac{1}{r}$ for $r > 0$. The first term here is lower-bounded using Proposition~\ref{prop:random-walk-majorization}:
\begin{equation}    \label{eqn:pca3}
    2\E_{\blambda} \sum_{i=1}^k \frac{s_{\blambda+e_i}(\alpha)}{\phantom{{}_{+e_i}}s_{\blambda}(\alpha)}
    \geq 2 \sum_{i=1}^k \alpha_i.
\end{equation}
As for the second term in~\eqref{eqn:pca2}, we use~\eqref{eqn:sw-probs} and the first formula in~\eqref{eqn:ssyt-formula} to compute
\begin{align}
    \E_{\blambda} \sum_{i=1}^k \frac{s_{\blambda+e_i}(\alpha)}{\phantom{{}_{+e_i}}s_{\blambda}(\alpha)}\frac{s_{\blambda+e_i}(1, \dots, 1)}{\phantom{{}_{+e_i}}s_{\blambda}(1, \dots, 1)}
&= \sum_{i=1}^k
\sum_{\lambda \vdash n} \symdim{\lambda}  s_{\lambda}(\alpha) \cdot \frac{s_{\lambda+e_i}(\alpha)}{\phantom{{}_{+e_i}}s_{\lambda}(\alpha)} \frac{\symdim{\lambda+e_i}(d+\lambda_i - i+1)}{\symdim{\lambda}(n+1)} \nonumber\\
&= \sum_{i=1}^k
\sum_{\lambda \vdash n} \symdim{\lambda+e_i} s_{\lambda+e_i}(\alpha)\cdot \frac{(d-i+\lambda_i+1)}{n+1}\nonumber\\
&\leq \sum_{i=1}^k \E_{\blambda' \sim \SW{n+1}{\alpha}} \frac{(d-i+\blambda'_i)}{n+1} \tag{by $\eqref{eqn:sw-probs}$ again}\\
&\leq \frac{1}{n+1} \cdot \E_{\blambda' \sim \SW{n+1}{\alpha}} \sum_{i=1}^k \blambda'_i + \frac{kd}{n} \nonumber\\
&\leq \sum_{i=1}^k \alpha_i + \frac{2\sqrt{2}k}{\sqrt{n}} + \frac{kd}{n},\label{eqn:pca4}
\end{align}
where the last step is by Lemma~\ref{lem:k-root-n}.
Combining \eqref{eqn:pca1}--\eqref{eqn:pca4} we get
\begin{equation*}
    \E_{\blambda, \bU} \| \bU \diag^{(k)}(\ulblam) \bU^\dagger -\rho\|_1
    \leq \left(1-\sum_{i=1}^k \alpha_i\right) + \sqrt{\frac{8kd}{n}} + \frac{2\sqrt{2}k}{\sqrt{n}} + \frac{kd}{n}
    \leq \sum_{i=k+1}^d \alpha_i + \sqrt{\frac{32kd}{n}} + \frac{kd}{n},
\end{equation*}
where the second inequality used $k \leq \sqrt{kd}$.  Finally, as the expectation is also trivially upper-bounded by~$2$, we may use $6\sqrt{r} \geq \min(2, \sqrt{32r} + r)$ (which holds for all $r \geq 0$) to conclude
\[
    \E_{\blambda, \bU} \| \bU \diag^{(k)}(\ulblam) \bU^\dagger -\rho\|_1
    \leq \sum_{i=k+1}^d \alpha_i + 6\sqrt{\frac{kd}{n}}. \qedhere
\]
\end{proof}

\section{Majorization for the RSK algorithm}                \label{sec:coupling}

In this section we prove Theorem~\ref{thm:coupling-main}.  The key to the proof will be the following strengthened version of the $d = 2$ case, which we believe is of independent interest.
\begin{theorem}                                     \label{thm:coupling-ssLIS}
    Let $0 \leq p, q \leq 1$ satisfy $|q - \frac12| \geq |p - \frac12|$; in other words, the $q$-biased probability distribution $(q,1-q)$ on $\{1,2\}$ is ``more extreme'' than the $p$-biased distribution $(p,1-p)$. 
    Then for any $n \in \N$ there is a coupling $(\bw, \bx)$ of the $p$-biased distribution on $\{1,2\}^n$ and the $q$-biased distribution on $\{1,2\}^n$ such that for all $1 \leq i \leq j \leq n$ we have $\LIS(\bx[i \dd j]) \geq \LIS(\bw[i \dd j])$ always.
\end{theorem}

We now show how to prove Theorem~\ref{thm:coupling-main} given Theorem~\ref{thm:coupling-ssLIS}.  Then in the following subsections we will prove Theorem~\ref{thm:coupling-ssLIS}.

\begin{proof}[Proof of Theorem~\ref{thm:coupling-main} given Theorem~\ref{thm:coupling-ssLIS}.]  A classic result of Muirhead~\cite{Mui02} (see also~\cite[B.1~Lemma]{MOA11}) says that $\beta \succ \alpha$ implies there is a sequence $\beta = \gamma_0 \succ \gamma_1 \succ \cdots \succ \gamma_t = \alpha$ such $\gamma_i$ and $\gamma_{i+1}$ differ in at most~$2$ coordinates.  Since the $\unrhd$ relation is transitive, by composing couplings it suffices to assume that $\alpha$ and $\beta$ themselves differ in at most two coordinates. Since the Schur--Weyl distribution is symmetric with respect to permutations of~$[d]$, we may assume that these two coordinates are~$1$ and~$2$.  Thus we may assume $\alpha = (\alpha_1, \alpha_2, \beta_3, \beta_4, \dots, \beta_d)$, where $\alpha_1 + \alpha_2 = \beta_1 + \beta_2$ and $\alpha_1, \alpha_2$ are between $\beta_1,\beta_2$.

We now define the coupling $(\blambda, \bmu)$ as follows:  We first choose a string $\bz \in (\{\ast\} \cup \{3, 4, \dots, d\})^n$ according to the product distribution in which symbol~$j$ has probability $\beta_j$ for $j \geq 3$ and symbol~$\ast$ has the remaining probability $\beta_1+\beta_2$.  Let $\bn_\ast$ denote the number of $\ast$'s in~$\bz$.  Next, we use Theorem~\ref{thm:coupling-ssLIS} to choose coupled strings $(\bw,\bx)$ with the $p$-biased distribution on $\{1,2\}^{\bn_\ast}$ and the $q$-biased distribution on $\{1,2\}^{\bn_\ast}$ (respectively), where $p = \frac{\alpha_1}{\beta_1 + \beta_2}$ and $q = \frac{\beta_1}{\beta_1 + \beta_2}$.  Note indeed that $|q - \frac12| \geq |p - \frac12|$, and hence $\LIS(\bx[i \dd j]) \geq \LIS(\bw[i \dd j])$ for all $1 \leq i \leq \bn_{\ast}$.  Now let ``$\bz \cup \bw$'' denote the string in~$[d]^n$ obtained by filling in the $\ast$'s in $\bz$ with the symbols from $\bw$, in the natural left-to-right order; similarly define ``$\bz \cup \bx$''.  Note that $\bz \cup \bw$ is distributed according to the product distribution $\alpha^{\otimes n}$ and likewise for $\bz \cup \bx$ and $\beta^{\otimes n}$.  Our final coupling is now obtained by taking $\blambda = \shRSK(\bz \cup \bw)$ and $\bmu = \shRSK(\bz \cup \bx)$. We need to show that $\bmu \unrhd \blambda$ always.

By Greene's Theorem, it suffices to show that if $s_1, \dots, s_k$ are disjoint increasing subsequences in $\bz \cup \bw$ of total length~$S$, we can find~$k$ disjoint increasing subsequences $s'_1, \dots, s'_k$ in $\bz \cup \bx$ of total length at least~$S$.   We first dispose of some simple cases. If none of $s_1, \dots, s_k$ contains any $1$'s or~$2$'s, then we may take $s'_i = s_i$ for $i \in [k]$, since these subsequences all still appear in $\bz \cup \bx$. The case when exactly one of $s_1, \dots, s_k$ contains any $1$'s or $2$'s is also easy.  Without loss of generality, say that $s_k$ is the only subsequence containing $1$'s and $2$'s.  We may partition it as $(t,u)$, where $t$ is a subsequence of $\bw$ and $u$ is a subsequence of the non-$\ast$'s in~$\bz$ that follow $\bw$.  Now let $t'$ be the longest increasing subsequence in~$\bx$.  As $t$ is an increasing subsequence of $\bw$, we know that $t'$ is at least as long as $t$.  Further, $(t', u)$ is an increasing subsequence in $\bz \cup \bx$.  Thus we may take $s'_i = s_i$ for $i < k$, and $s'_k = (t',u)$.

We now come to the main case, when at least two of $s_1, \dots, s_k$ contain $1$'s and/or $2$'s.  Let's first look at the position $j \in [n]$ of the rightmost~$1$ or~$2$ among $s_1, \dots, s_k$.  Without loss of generality, assume it occurs in~$s_k$.  Next, look at the position $i \in [n]$ of the rightmost~$1$ or~$2$ among $s_1, \dots, s_{k-1}$.  Without loss of generality, assume it occurs in $s_{k-1}$.  We will now modify the subsequences $s_1, \dots, s_k$ as follows:
\begin{itemize}
    \item all $1$'s and $2$'s are deleted from $s_1, \dots, s_{k-2}$ (note that these all occur prior to position~$i$);
    \item $s_{k-1}$ is changed to consist of all the $2$'s within $(\bz \cup \bw)[1 \dd i]$;
    \item the portion of $s_{k}$ to the right of position~$i$ is unchanged, but the preceding portion is changed to consist of all the $1$'s within $(\bz \cup \bw)[1 \dd i]$.
\end{itemize}
It is easy to see that the new $s_1, \dots, s_k$ remain disjoint subsequences of $\bz \cup \bw$, with total length at least~$S$.  We may also assume that the portion of $s_k$ between positions $i+1$ and $j$ consists of a longest increasing subsequence of~$\bw$.

Since the subsequences $s_1, \dots, s_{k-2}$ don't contain any $1$'s or $2$'s, they still appear in $\bz \cup \bx$, and we may take these as our $s'_1, \dots, s'_{k-2}$.  We will also define $s'_{k-1}$ to consist of all $2$'s within $(\bz \cup \bx)[1 \dd i]$. Finally, we will define $s'_k$ to consist of all $1$'s within $(\bz \cup \bz)[1 \dd i]$, followed by the longest increasing subsequence of $\bx$ occurring within positions $(i+1) \dd j$ in $\bz \cup \bx$, followed by the portion of $s_k$ to the right of position~$j$ (which does not contain any $1$'s or $2$'s and hence is still in $\bz \cup \bx$).  It is clear that $s'_1, \dots, s'_k$ are indeed disjoint increasing subsequences of $\bz \cup \bx$.  Their total length is the sum of four quantities:
\begin{itemize}
    \item the total length of $s_1, \dots, s_{k-2}$;
    \item the total number of $1$'s and $2$'s within $(\bz \cup \bx)[1 \dd i]$;
    \item the length of the longest increasing subsequence of $\bx$ occurring within positions $(i+1) \dd j$ in $\bz \cup \bx$;
    \item the length of the portion of $s_k$ to the right of position~$j$.
\end{itemize}
By the coupling property of $(\bw, \bx)$, the third quantity above is at least the length of the longest increasing subsequence of $\bw$ occurring within positions $(i+1) \dd j$ in $\bz \cup \bw$.   But this precisely shows that the total length of $s'_1, \dots, s'_k$ is at least that of $s_1, \dots, s_k$, as desired.
\end{proof}

\subsection{Substring-LIS-dominance: RSK and Dyck paths}
In this subsection we make some preparatory definitions and observations toward proving Theorem~\ref{thm:coupling-ssLIS}.  We begin by codifying the key property therein.

\begin{definition}
    Let $w, w' \in \calA^n$ be strings of equal length.  We say $w'$ \emph{substring-LIS-dominates}~$w$, notated $w' \ssLIS w$, if $\LIS(w'[i\dd j]) \geq \LIS(w[i\dd j])$ for all $1 \leq i \leq j \leq n$.  (Thus the coupling in Theorem~\ref{thm:coupling-ssLIS} satisfies $\bw \ssLIS \bv$ always.)  The relation $\ssLIS$ is reflexive and transitive.  If we have the substring-LIS-dominance condition just for $i = 1$ we say that $w'$ \emph{prefix-LIS-dominates}~$w$. If we have it just for $j = n$ we say that $w'$~\emph{suffix-LIS-dominates}~$w$.
\end{definition}

\begin{definition}
    For a string $w \in \calA^n$ we write $\behead(w)$ for $w[2 \dd n]$ and $\curtail(w)$ for $w[1 \dd n-1]$.
\end{definition}

\begin{remark}  \label{rem:recursive-ssLIS}
    We may equivalently define substring-LIS-dominance recursively, as follows.  If $w'$ and $w$ have length~$0$ then $w' \ssLIS w$.  If $w'$ and $w$ have length $n > 0$, then $w' \ssLIS w$ if and only if $\LIS(w') \geq \LIS(w)$ and $\behead(w') \ssLIS \behead(w)$ and $\curtail(w') \ssLIS \curtail(w)$.  By omitting the second/third condition we get a recursive definition of prefix/suffix-LIS-dominance.
\end{remark}

\begin{definition}
    Let $Q$ be a (nonempty) standard Young tableau.  We define $\curtail(Q)$ to be the standard Young tableau obtained by deleting the box with maximum label from~$Q$.
\end{definition}
The following fact is immediate from the definition of the RSK correspondence:
\begin{proposition}                                     \label{prop:curtail-rsk}
    Let $w \in \calA^n$ be a nonempty string.  Suppose $\RSK(w) = (P,Q)$ and $\RSK(\curtail(w)) = (P',Q')$.  Then $Q' = \curtail(Q)$.
\end{proposition}
The analogous fact for \beheading is more complicated.
\begin{definition}                                      \label{def:behead-tableau}
    Let $Q$ be a (nonempty) standard Young tableau.  We define $\behead(Q)$ to be the standard Young tableau obtained by deleting the top-left box of~$Q$, sliding the hole outside of the tableau according to jeu de taquin (see, e.g.,~\cite{Ful97,Sag01}), and then decreasing all entries by~$1$.  (The more traditional notation for $\behead(Q)$ is $\Delta(Q)$.)
\end{definition}
The following fact is due to~\cite{Sch63}; see~\cite[Proposition~3.9.3]{Sag01} for an explicit proof.\footnote{Technically, therein it is proved only for strings with distinct letters.  One can recover the result for general strings in the standard manner; if the letters $w_i$ and $w_j$ are equal we break the tie by using the order relation on $i,j$.  See also~\cite[Lemma]{Lee13}.}
\begin{proposition}                                        \label{prop:behead-rsk}
    Let $w \in \calA^n$ be a nonempty string. Suppose $\RSK(w) = (P,Q)$ and $\RSK(\behead(w)) = (P',Q')$.  Then $Q' = \behead(Q)$.
\end{proposition}

\begin{proposition}                                     \label{prop:ssLIS-from-recording}
    Let $w, w' \in \calA^n$ be strings of equal length and write $\RSK(w) = (P,Q)$, $\RSK(w') = (P',Q')$. Then whether or not $w' \ssLIS w$ can be determined just from the recording tableaus $Q'$ and~$Q$.
\end{proposition}
\begin{proof}
    This follows from the recursive definition of $\ssLIS$ given in Remark~\ref{rem:recursive-ssLIS}: whether $\LIS(w') \geq \LIS(w)$ can be determined by checking whether the first row of~$Q'$ is at least as long as the first row of~$Q$; the recursive checks can then be performed with the aid of Propositions~\ref{prop:curtail-rsk},~\ref{prop:behead-rsk}.
\end{proof}
\begin{definition}                                        \label{def:ssLIS-tableau}
    In light of Proposition~\ref{prop:ssLIS-from-recording} we may define the relation $\ssLIS$ on standard Young tableaus.
\end{definition}
\begin{remark}                                                  \label{rem:visual-sslis-tableau}
    The simplicity of Proposition~\ref{prop:curtail-rsk} implies that it is very easy to tell, given $w, w' \in \calA^n$ with recording tableaus $Q$ and $Q'$, whether $w'$ suffix-LIS-dominates $w$.  One only needs to check whether $Q'_{1j} \leq Q_{1j}$ for all $j \geq 1$ (treating empty entries as~$\infty$).  On the other hand, it is not particularly easy to tell from $Q'$ and $Q$ whether $w'$ prefix-LIS-dominates $w$; one seems to need to execute all of the jeu de taquin slides.
\end{remark}

We henceforth focus attention on alphabets of size~$2$. Under RSK, these yield standard Young tableaus with at most $2$-rows.  (For brevity, we henceforth call these \emph{$2$-row Young tableaus}, even when they have fewer than $2$~rows.) In turn, $2$-row Young tableaus can be identified with Dyck paths (also known as ballot sequences).
\begin{definition}
    We define a \emph{Dyck path of length~$n$} to be a path in the $xy$-plane that starts from $(0,0)$, takes $n$ steps of the form $(+1,+1)$ (an \emph{upstep}) or $(+1,-1)$ (a \emph{downstep}), and never passes below the $x$-axis.  We say that the \emph{height} of a step~$s$, written $\high(s)$, is the $y$-coordinate of its endpoint; the \emph{(final) height} of a Dyck path~$W$, written $\high(W)$, is the height of its last step.  We do \emph{not} require the final height of a path to be~$0$; if it is we call the path \emph{complete}, and otherwise we call it \emph{incomplete}.  A \emph{return} refers to a point where the path returns to the $x$-axis; i.e., to the end of a step of height~$0$. An \emph{arch} refers to a minimal complete subpath of a Dyck path; i.e., a subpath between two consecutive returns (or between the origin and the first return).
\end{definition}

\begin{definition}              \label{def:dyck-ident}
    We identify each $2$-row standard Young tableau~$Q$ of size~$n$ with a Dyck path~$W$ of length~$n$.  The identification is the standard one: reading off the entries of $Q$ from $1$ to $n$, we add an upstep to $W$ when the entry is in the first row and a downstep when it is in the second row.  The fact that this produces a Dyck path (i.e., the path does not pass below the $x$-axis) follows from the standard Young tableau property.  Note that the final height of $W$  is the difference in length between $Q$'s two rows.  We also naturally extend the terminology ``return'' to $2$-row standard Young tableaus~$Q$: a \emph{return} is a second-row box labeled $2j$ such that boxes in~$Q$ labeled $1, \dots, 2j$ form a rectangular $2 \times j$ standard Young tableau.
\end{definition}

\begin{definition}                                          \label{def:ssLIS-dyck}
    In light of Definition~\ref{def:ssLIS-tableau} and the above identification, we may define the relation $\ssLIS$ on Dyck paths.
\end{definition}

Of course, we want to see how \beheading and \curtailment apply to Dyck paths.  The following fact is immediate:
\begin{proposition}                                     \label{prop:curtail-dyck}
    If $W$ is the Dyck path corresponding to a nonempty $2$-row standard Young tableau~$Q$, then the Dyck path $W'$ corresponding to $\curtail(Q)$ is formed from $W$ by deleting its last segment.  We write $W' = \curtail(W)$ for this new path.
\end{proposition}
Again, the case of \beheading is more complicated.  We first make some definitions.
\begin{definition}
    \emph{Raising} refers to converting a downstep in a Dyck path to an upstep; note that this increases the Dyck path's height by~$2$.  Conversely, \emph{lowering} refers to converting an upstep to a downstep. Generally, we only allow lowering when the result is still a Dyck path; i.e., never passes below the $x$-axis.
\end{definition}
\begin{proposition}                                     \label{prop:behead-dyck}
    Let $Q$ be a nonempty $2$-row standard Young tableau, with corresponding Dyck path $W$.  Let $W'$ be the Dyck path corresponding to $\behead(Q)$.  Then $W'$ is formed from $W$ as follows: First, the initial step of $W$ is deleted (and the origin is shifted to the new initial point).  If $W$ had no returns then the operation is complete and $W'$ is the resulting Dyck path.  Otherwise, if $W$ had at least one return, then in the new path~$W'$ that step (which currently goes below the $x$-axis) is raised.  In either case, we write $W' = \behead(W)$ for the resulting path.
\end{proposition}
\begin{proof}
    We use Definitions~\ref{def:behead-tableau} and~\ref{def:dyck-ident}.  Deleting the top-left box of $Q$ corresponds to deleting the first step of~$W$, and decreasing all entries in~$Q$ by~$1$ corresponds to shifting the origin in~$W$.  Consider now the jeu de taquin slide in~$Q$.  The empty box stays in the first row until it first reaches a position~$j$ such that $Q_{1,j+1} > Q_{2,j}$ --- if such a position exists.  Such a position does exist if and only if~$Q$ contains a return (with box $(2,j)$ being the first such return).  If $Q$ (equivalently,~$W$) has no return then the empty box slides out of the first row of~$Q$, and indeed this corresponds to making no further changes to~$W$.  If $Q$ has its first return at box $(2,j)$, this means the jeu de taquin will slide up the box labeled~$2j$ (corresponding to raising the first return step in~$W$); then all remaining slides will be in the bottom row of~$Q$, corresponding to no further changes to~$W$.
\end{proof}

\begin{remark}                                          \label{rem:visual-sslis-dyck}
    Similar to Remark~\ref{rem:visual-sslis-tableau}, it is easily to ``visually'' check the suffix-LIS-domination relation for Dyck paths: $W'$ suffix-LIS-dominates~$W$ if and only if $W'$ is at least as high as $W$ throughout the length of both paths.  On the other hand, checking the full substring-LIS-domination relation is more involved; we have $W' \ssLIS W$ if and only if for any number of simultaneous \beheadings to $W'$ and $W$, the former path always stays at least as high as the latter.
\end{remark}

Finally, we will require the following definition:
\begin{definition}
    A \emph{hinged range} is a sequence $(R_0, s_1, R_1, s_2, R_2, \dots, s_k, R_k)$ (with $k \geq 0$), where each $s_i$ is a step (upstep or downstep) called a \emph{hinge} and each $R_i$ is a Dyck path (possibly of length~$0$) called a \emph{range}.  The ``internal ranges'' $R_1, \dots, R_{k-1}$ are required to be complete Dyck paths; the ``external ranges'' $R_0$ and $R_k$ may be incomplete.

    We may identify the hinged range with the path formed by concatenating its components; note that this need not be a Dyck path, as it may pass below the origin.

    If $H$ is a hinged range and $H'$ is formed by raising zero or more of its hinges (i.e., converting downstep hinges to upsteps), we say that $H'$ is a \emph{raising} of $H$ or, equivalently, that $H$ is a \emph{lowering} of $H'$. 
    We call a hinged range \emph{fully lowered} (respectively, \emph{fully raised}) if all its hinges are downsteps (respectively, upsteps).
\end{definition}

\subsection{A bijection on Dyck paths}

\begin{theorem}                                     \label{thm:dyck-bij}
    Fix integers $n \geq 2$ and $1 \leq \lambda_2 \leq \lfloor \frac{n}{2}\rfloor$.
    Define
    \begin{align*}
        \calW = \bigl\{ (W, s_1) : &\textnormal{ $W$ is a length-$n$ Dyck path with exactly $\lambda_2$ downsteps;} \\
                                               &\textnormal{ $s_1$ is a downstep in~$W$}\bigr\}
    \end{align*}
    and
    \begin{align*}
        \calW' = \smash{\bigcup_{k=1}^{\lambda_2}}
                     \bigl\{ (W', s'_1) : &\textnormal{ $W'$ is a length-$n$ Dyck path with exactly $\lambda_2-k$ downsteps;} \\
                                                 &\textnormal{ $s'_1$ is an upstep in~$W'$ with $k+1 \leq \high(s'_1) \leq \high(W')-k+1$;} \\
                                                 &\textnormal{ $s'_1$ is the rightmost upstep in~$W'$ of its height}\bigr\}.
    \end{align*}
    Then there is an explicit bijection $f \co \calW \to \calW'$ such that whenever $f(W,s_1) = (W',s'_1)$ it holds that $W' \ssLIS W$.
\end{theorem}
\begin{remark}                                       \label{rem:dyck-bij}
    Each length-$n$ Dyck path with exactly $\lambda_2$ downsteps occurs exactly $\lambda_2$ times in~$\calW$. Each length-$n$ Dyck path with strictly fewer than $\lambda_2$ downsteps occurs exactly $n-2\lambda_2+1$ times in~$\calW'$.
\end{remark}
\begin{proof}[Proof of Theorem~\ref{thm:dyck-bij}]
    Given any $(W,s_1) \in \calW$, we define $f$'s value on it as follows. Let $s_2$ be the first downstep following $s_1$ in $W$ having height $\high(s_1)-1$; let $s_3$ be the first downstep following $s_2$ in $W$ following $s_2$ having height $\high(s_2) - 1$; etc., until reaching downstep $s_k$ having no subsequent downstep of smaller height.  Now decompose $W$ as a (fully lowered) hinged range $H = (R_0, s_1, R_1, \dots, s_k, R_k)$.  Let $H' = (R_0', s_1', R_1', \dots, s_k', R_k')$ be the fully raised version of~$H$ (where each $R_j'$ is just $R_j$ and each $s'_j$ is an upstep).  Then $f(W,s_k)$ is defined to be $(W', s_1')$, where $W'$ is the Dyck path corresponding to $H'$.

    First we check that indeed $(W',s_1') \in \calW'$.  As $W'$ is formed from $W$ by $k$ raisings, it has exactly $\lambda_2-k$ downsteps. 
    Since $\high(s_k) \geq 0$ it follows that $\high(s_1) \geq k-1$ and hence $\high(s'_1) \geq k+1$.  On the other hand, $\high(s'_1) + (k-1) = \high(s'_k) \leq \high(W')$ and so $\high(s'_1) \leq \high(W') -k+1$.  Finally, $s'_1$ is the rightmost upstep in $W'$ of its height because $H'$ is fully raised.

    To show that $f$ is a bijection, we will define the function $g : \calW' \to \calW$ that will evidently be $f$'s inverse. Given any $(W',s_1') \in \calW$, with $W'$ having exactly $\lambda_2-k$ downsteps, we define $g$'s value on it as follows.  Let $s'_2$ be the \emph{last} (rightmost) upstep following $s'_1$ in $W'$ having height $\high(s'_1)+1$; let $s'_3$ be the last upstep following $s'_2$ in $W'$ having height $\high(s'_2)+1$; etc., until $s'_k$ is defined.  That this $s'_k$ indeed exists follows from the fact that $\high(s_1') \leq \high(W')-k+1$.  Now decompose $W'$ as a (fully raised) hinged range $H' = (R'_0, s_1', R_1', \dots, s_k', R_k')$.  The fact that $R_k'$ is a Dyck path (i.e., does not pass below its starting height) again follows from the fact that $\high(s_k') = \high(s_1')+k-1 \leq \high(W')$. Finally, let $H = (R_0, s_1, R_1, \dots, s_k, R_k)$ be the fully lowered version of~$H'$, and $W$ the corresponding path.  As $W$ has exactly $\lambda_2$ downsteps, we may define $g(W',s_1') = (W,s_1)$ provided $W$ is indeed a Dyck path.  But this is the case, because the lowest point of~$W$ occurs at the endpoint of~$s_k$, and $\high(s_k) = \high(s_1) - k+1 = \high(s'_1)-2 - k+1 = \high(s'_1) - k-1 \geq 0$ since $\high(s'_1) \geq k+1$.

    It is fairly evident that $f$ and $g$ are inverses.  The essential thing to check is that the sequence $s_1, \dots, s_k$ determined from $s_1$ when computing $f(W,s_1)$ is ``the same'' (up to raising/lowering) as the sequence $s_1', \dots, s'_{k'}$ determined from $s_1'$ in computing $g(W',s'_1)$, and vice versa.  The fact that the sequences have the same \emph{length} follows, in the $g \circ f = id$ case, from the fact that $\high(W') = \high(W) + 2k$; it follows, in the $f \circ g = id$ case, from the fact that $R_k'$ is a Dyck path.  The fact that the hinges have the same identity is evident from the nature of fully raising/lowering hinged ranges.

    It remains to show that if $f(W,s_1) = (W',s'_1)$ then $W' \ssLIS W$. Referring to Remark~\ref{rem:visual-sslis-dyck}, we need to show that if $W'$ and $W$ are both simultaneously \beheaded some number of times~$b$, then in the resulting paths, $W'$ is at least as high as $W$ throughout their lengths.  In turn, this is implied by the following more general statement:
    \begin{claim}
        After $b$ \beheadings, $W'$ and $W$ may be expressed as hinged ranges $H' = (R_0, s'_1, R_1, \dots, s'_k, R_k)$ and $H = (R_0, s_1, R_1, \dots, s_k, R_k)$ (respectively) such that $H'$ is the fully raised version of~$H$ (i.e., each $s'_j$ is an upstep).
    \end{claim}
    \noindent (Note that we do not necessarily claim that $H$ is the fully lowered version of $H'$.)

    The claim can be proved by induction on~$b$.  The base case $b = 0$ follows by definition of~$f$.  Throughout the induction we may assume that the common initial Dyck path $R_0$ is nonempty, as otherwise $s_1$ must be an upstep, in which case we can redefine the common initial Dyck path of $W$ and $W'$ to be $(s_1, R_1) = (s'_1, R_1)$.  

    We now show the inductive step. Assume $W'$ and $W$ are nonempty paths as in the claim's statement, with $R_0$ nonempty. 
    Suppose now that $W'$ and $W$ are simultaneously \beheaded.  The first step of $W'$ and $W$ (an upstep belonging to $R_0$) is thus deleted, and the origin shifted.  If $R_0$ contained a downstep to height~$0$ then the first such downstep is raised in both $\behead(W')$ and $\behead(W)$ and the inductive claim is maintained.  Otherwise, suppose $R_0$ contained no downsteps to height~$0$.  It follows immediately that $W'$ originally had no returns to height~$0$ at all; hence the \beheading of $W'$ is completed by the deletion of its first step.  It may also be that $W$ had no returns to height~$0$ at all; then the \beheading of $W$ is also completed by the deletion of its first step and the induction hypothesis is clearly maintained.  On the other hand, $W$ \emph{may} have had some downsteps to~$0$ within~$(s_1, R_1, \dots, s_k, R_k)$.  In this case, the first (leftmost) such downstep must occur at one of the hinges~$s_j$, and the \beheading of $W$ is completed by raising this hinge.  The inductive hypothesis is therefore again maintained.  This completes the induction.
\end{proof}

We derive an immediate corollary, after introducing a bit of notation:
\newcommand{\SYTeq}[2]{\SYT_{#1}({=}#2)}
\newcommand{\SYTleq}[2]{\SYT_{#1}({\leq} #2)}
\begin{definition}
    We write $\SYTeq{n}{\lambda_2}$ (respectively, $\SYTleq{n}{\lambda_2}$) for the set of $2$-row standard Young tableaus of size~$n$ with exactly (respectively, at most) $\lambda_2$ boxes in the second row.
\end{definition}
\begin{corollary}                                       \label{cor:tableau-bij1}
    For any integers $n \geq 2$ and $0 \leq \lambda_2 \leq \lfloor \frac{n}{2} \rfloor$, there is a coupling $(\bQ, \bQ')$ of the uniform distribution on $\SYTeq{n}{\lambda_2}$  and the uniform distribution on $\SYTleq{n}{\lambda_2-1}$ such that $\bQ' \ssLIS \bQ$ always.
\end{corollary}
\begin{proof}
    Let $(\bW, \bs_1)$ be drawn uniformly at random from the set~$\calW$ defined in Theorem~\ref{thm:dyck-bij}, and let $(\bW',\bs'_1) = f(\bW,\bs_1)$. Let $\bQ \in \SYTeq{n}{\lambda_2}$, $\bQ' \in \SYTleq{n}{\lambda_2\!-\!1}$ be the $2$-row standard Young tableaus identified with $\bW$, $\bW'$ (respectively).  Then Theorem~\ref{thm:dyck-bij} tells us that $\bQ' \ssLIS \bQ$ always, and Remark~\ref{rem:dyck-bij} tells us that $\bQ$ and $\bQ'$ are each uniformly distributed.
\end{proof}
\begin{corollary}                                       \label{cor:tableau-bij2}
    For any integers $n \geq 0$ and $0 \leq \lambda_2' \leq \lambda_2 \leq \lfloor \frac{n}{2} \rfloor$, there is a coupling $(\bQ, \bQ')$ of the uniform distribution on $\SYTleq{n}{\lambda_2}$  and the uniform distribution on $\SYTleq{n}{\lambda_2'}$ such that $\bQ' \ssLIS \bQ$ always.
\end{corollary}
\begin{proof}
    The cases $n < 2$ and $\lambda_2' = \lambda_2$ are trivial, so we may assume $n \geq 2$ and $0 \leq \lambda_2' < \lambda_2 \leq \lfloor \frac{n}{2} \rfloor$.  By composing couplings and using transitivity of~$\ssLIS$, it suffices to treat the case $\lambda_2' = \lambda_2 -1$.  But the uniform distribution on $\SYTleq{n}{\lambda_2}$ is a mixture of (a)~the uniform distribution on $\SYTeq{n}{\lambda_2}$, (b)~the uniform distribution on $\SYTleq{n}{\lambda_2-1}$; and these can be coupled to $\SYTleq{n}{\lambda_2-1}$ under the $\ssLIS$ relation using (a)~Corollary~\ref{cor:tableau-bij1}, (b)~the identity coupling.
\end{proof}

Before giving the next corollary, we have a definition.
\begin{definition}                          \label{def:Ham-SymHam}
    Let $\calA$ be any $2$-letter alphabet. We write $\Ham{\calA}{n}{k}$ for the set of length-$n$ strings over~$\calA$ with exactly $k$ copies of the larger letter, and we write $\SymHam{\calA}{n}{k} = \Ham{\calA}{n}{k} \cup \Ham{\calA}{n}{n-k}$.
\end{definition}
\begin{corollary}                                       \label{cor:string-bij}
    For $\calA$ a $2$-letter alphabet and integers $0 \leq k' \leq k \leq \lfloor \frac{n}{2} \rfloor$, there is a coupling $(\bw,\bw')$ of the uniform distribution on $\SymHam{\calA}{n}{k}$ and the uniform distribution on $\SymHam{\calA}{n}{k'}$ such that $\bw' \ssLIS \bw$ always.
\end{corollary}
\begin{proof}
    We first recall that if $\bx \sim \Ham{\calA}{n}{k}$ is uniformly random and $(\bP,\bQ) = \RSK(\bx)$, then the recording tableau $\bQ$ is uniformly random on $\SYTleq{n}{k}$. This is because for each possible recording tableau $Q \in \SYTleq{n}{k}$ there is a unique insertion tableau~$P$ of the same shape as~$Q$ having exactly $k$ boxes labeled with the larger letter of~$\calA$. (Specifically, if $P \vdash (\lambda_1, \lambda_2)$, then the last $k-\lambda_2$ boxes of $P$'s first row, and all of the boxes of $P$'s second row, are labeled with $A$'s larger letter.)  It follows that the same is true if $\bx \sim \SymHam{\calA}{n}{k}$ is uniformly random.  But now the desired coupling follows from Corollary~\ref{cor:tableau-bij2} (recalling Definition~\ref{def:ssLIS-tableau}).
\end{proof}

In fact, Corollary~\ref{cor:string-bij} is fundamentally stronger than our desired Theorem~\ref{thm:coupling-ssLIS}, as we now show:

\begin{proof}[Proof of Theorem~\ref{thm:coupling-ssLIS}.]
    For $r \in [0,1]$, suppose we draw an $r$-biased string $\by \in \{1,2\}^n$ and define the random variable $\bj$ such that $\by \in \SymHam{\{1,2\}}{n}{\bj}$.  (Note that given $\bj$, the string $\by$ is uniformly distributed on $\SymHam{\{1,2\}}{n}{\bj}$.)  Write $L_r(\ell)$ for the cumulative distribution function of~$\bj$; i.e., $L_r(\ell) = \Pr[\by \in \cup_{j \leq \ell} \SymHam{\{1,2\}}{n}{j}]$, where $\by$ is $r$-biased.

    \medskip

    \emph{Claim:}
        $L_q(\ell) \geq L_p(\ell)$ for all $0 \leq \ell \leq \lfloor \frac{n}{2} \rfloor$.

    \medskip

    Before proving the claim, let us show how it is used to complete the proof of Theorem~\ref{thm:coupling-ssLIS}. We define the required coupling $(\bw,\bx)$ of $p$-biased and $q$-biased distributions as follows:  First we choose $\btheta \in [0,1]$ uniformly at random. Next we define $\bk$ (respectively, $\bk'$) to be the least integer such that $L_p(\bk) \geq \btheta$ (respectively, $L_q(\bk') \geq \btheta$); from the claim it follows that $\bk' \leq \bk$ always.  Finally, we let $(\bw, \bx)$ be drawn from the coupling on $\SymHam{\{1,2\}}{n}{\bk}$ and $\SymHam{\{1,2\}}{n}{\bk'}$ specified in Corollary~\ref{cor:string-bij}. Then as required, we have that $\bx' \ssLIS \bw$ always, and that $\bw$ has the $p$-biased distribution and $\bx$ has the $q$-biased distribution.

    It therefore remains to prove the claim.  We may exclude the trivial cases $\ell = \frac{n}{2}$ or $q \in \{0,1\}$, where $L_q(\ell) = 1$. Also, since $L_r(\ell) = L_{1-r}(\ell)$ by symmetry, we may assume $0 < q \leq p \leq \frac12$. Thus it suffices to show that $\frac{d}{dr} L_r(\ell) \leq 0$ for $0 < r \leq \frac12$. Letting $\bh$ denote the ``Hamming weight'' (number of~$2$'s) in an $r$-biased random string on $\{1,2\}^n$, we have
    \begin{align*}
        L_r(\ell) &= \Pr[\bh \leq \ell] + \Pr[\bh \geq n-\ell] = 1- \Pr[\bh > \ell] + \Pr[\bh > n-\ell-1] \\
         \Rightarrow \frac{d}{dr} L_r(\ell) &= -\frac{d}{dr} \Pr[\bh > \ell] + \frac{d}{dr} \Pr[\bh > n-1-\ell].
    \end{align*}
    (The first equality used $\ell < \frac{n}{2}$.)  But it is a basic fact that $\frac{d}{dr} \Pr[\bh > t] = n \binom{n-1}{t}r^{t}(1-r)^{n-1-t}$. Thus
    \[
        \frac{d}{dr} L_r(\ell) = n \binom{n-1}{\ell}\left(- r^\ell(1-r)^{n-1-\ell}+r^{n-1-\ell}(1-r)^\ell\right),
    \]
    and we may verify this is indeed nonpositive:
    \[
        - r^\ell(1-r)^{n-1-\ell}+r^{n-1-\ell}(1-r)^\ell \leq 0  \iff 1 \leq \left(\tfrac{1-r}{r}\right)^{n-1-2\ell},
    \]
    which is true since $0 < r \leq \frac12$ and $n-1-2\ell \geq 0$ (using $\ell < \frac{n}{2}$ again).
\end{proof}

\bibliographystyle{alpha}
\bibliography{wright}

\end{document}